%% file: est_after_sel_PartI_18.tex
\documentclass[journal]{IEEEtran}
 \newcommand{\ud}{\,\mathrm{d}}

\input{myshorts}

\usepackage{amssymb, amsmath}
\newtheorem{theorem}{Theorem}
\newtheorem{lemma}{Lemma}
 \newtheorem{definition}{Definition}
  \newtheorem{proposition}{Proposition}
	 
 \usepackage[dvips]{color} 
\usepackage{lscape}
\input{psfig.sty}
 \usepackage{epsfig}
\usepackage[usenames,dvipsnames]{pstricks}

\def\ie{{\it i.e.,\ \/}}

\title{Estimation after Parameter Selection: Performance Analysis and Estimation Methods}
\author{Tirza~Routtenberg,~\IEEEmembership{Member,~IEEE},
        and~Lang~Tong,~\IEEEmembership{Fellow,~IEEE}
\thanks{T. Routtenberg and L. Tong are with School of Electrical and Computer Engineering, Cornell University, Ithaca, NY 14853, United States, Email: \{tsr43,lt35\}@cornell.edu.}
\thanks{This work is supported in part by the National Science Foundation under
Grant CNS-1135844.}
}

\begin{document}
\maketitle
\begin{abstract}
In many practical parameter estimation problems,
  prescreening and parameter selection are performed prior to estimation.
In this paper, we consider
the problem of estimating a preselected  unknown deterministic parameter chosen from a parameter set
based on observations according to a predetermined selection rule, $\Psi$.
The data-based parameter selection process may impact the subsequent estimation by introducing a selection bias 
 and creating coupling between decoupled parameters. This paper introduces a post-selection mean squared error (PSMSE) criterion as a performance measure.  
  A corresponding
 Cram$\acute{\text{e}}$r-Rao-type bound on the PSMSE of any $\Psi$-unbiased estimator is derived,
where the  $\Psi$-unbiasedness is in the
 Lehmann-unbiasedness  sense.
The post-selection maximum-likelihood (PSML) estimator is presented. 
It is proved that if there exists an $\Psi$-unbiased estimator that achieves the $\Psi$-Cram$\acute{\text{e}}$r-Rao bound (CRB),
\ie an $\Psi$-efficient estimator,
 then  it is produced by the PSML estimator.
In addition,  iterative  methods are developed for the practical implementation of the PSML estimator.
Finally, the proposed $\Psi$-CRB and  PSML estimator are
 examined  in estimation after parameter selection with  different distributions.
\end{abstract}

\section{Introduction}
\label{sec:intro}
Statistical inference on multiple parameters often involves a preliminary data-driven parameter selection stage.
In mathematical statistics literature,
 estimation after parameter selection refers to the problem  in which  estimation is  performed only after a specific population,  related to a specific parameter  has been selected from a set of possible independent populations.
The population selection is based on some predetermined data-based
selection rule, $\Psi$,  where $\Psi$ may not be optimal in any sense.
In cognitive radio communications \cite{Haykin_cognitive_radio}, for example, the parameters of a channel are estimated only after the channel has been identified in the white space,
 often
 thresholding on the empirical signal to noise ratio as a selection criterion. 
 In medical diagnoses, a special test is administered only after  other preliminary tests indicate that a patient may have contracted a certain disease.   
  Other applications include multiple radar subset selection problems  \cite{Petropulu}, medical experiments
	\cite{BIMJ200810442}, genetic studies \cite{Zollner2007}, and estimation in wireless sensor networks after  sensor node selection 
	\cite{Nehorai_Zhao}.

Despite the importance of estimation after parameter selection,
the impact of selection procedure on the fundamental limits of estimation performance for general parametric models is not well understood.
  It is known that the selection process affect the statistical properties of the subsequent estimation  
\cite{mukhopadhyay1994multistage}. In particular,
 the bias and mean squared error (MSE) criterion are inappropriate (e.g. \cite{Sarkadi}, \cite{Putter_Rubinstein}) and the conventional Cram$\acute{\text{e}}$r-Rao  bound (CRB) is unsuited since it does not take 
the prescreening process into account. 
In addition, the selection may create  coupling between originally decoupled parameters 
 and it usually induces a bias, or
``winner’s curse" \cite{Zollner2007}, on any estimator of the selected unknown parameter.
 For example, for the 
exponential family of distributions,
 no unbiased estimator  exists for classical estimation after  selection with 
independent  population and
 a single sampling stage and 
 data-based selection rules
{\cite{Putter_Rubinstein,Cohen_Sackrowitz,Vellaisamy2009}}.

\subsection{Summary of results}
In this work, we are interested in the problem of estimation after parameter selection,
 \ie estimating a subset of parameters after they are selected based on a data-based selection rule. 
 This problem is a generalization of the classical estimation after  selection problem \cite{mukhopadhyay1994multistage},
where each parameter is associated with
a specific non-overlapped set of observations, named a {\em{population}}, and
the populations are assumed to be independent.
Another special case of the model considered here is the problem of estimation in the presence of nuisance parameters 
\cite{Kay_estimation}, \cite{Andrea_Mengali_Reggiannini_1994}.  In such a problem the parameter of interest is chosen in advance {\em{independent of data}}.    

In order to characterize the estimation performance of the selected parameter, 
we introduce   the post-selection MSE (PSMSE) criterion
and  the concept of $\Psi$-unbiasedness by using the  non-Bayesian Lehmann-unbiasedness definition \cite{Lehmann}.
Then we develop the appropriate CRB-type bound on the PSMSE of any $\Psi$-unbiased estimator.
In addition, we   present the  post-selection maximum-likelihood (PSML) estimator, which is  
the corresponding 
 maximum-likelihood (ML) estimator for estimation after parameter selection problems.
 We show that if  an $\Psi$-unbiased estimator exists that achieves the $\Psi$-CRB,
 it is produced by the PSML estimator.
 We further develop iterative  methods for the practical implementation of the PSML estimator.
Finally, the proposed
 $\Psi$-CRB and PSML estimator  are examined on
 uniform, exponential, and Gaussian distributions with the 
sample mean selection (SMS) rule.
%


\subsection{Related works}
The earliest works on classical estimation after selection with independent populations are by Sarkadi, Putter, and Rubinstein 
in  \cite{Sarkadi} and \cite{Putter_Rubinstein}.
These works,  as well as studies that appear in mathematical statistics literature, 
assume {\em{random}} unknown parameters  
 and  show  that  no unbiased estimator exists for  independent Gaussian  populations.
In mathematical statistics literature, estimation after selection with independent populations has received considerable attention over the years,
 where most of the work 
is restricted to specific parametric models,
such as the Gaussian {\cite{BIMJ200810442,Cohen_Sackrowitz,Gibbons,Lu_Sun_Wu_2013}},
Gamma {\cite{Vellaisamy_Sharma_1988,Misra_2006}}, and
uniform \cite{song_uniform} models.
Several estimation methods have been proposed to reduce the {\em{selection bias}}
  by employing various iterative methods for bias correction 
(e.g. \cite{WHITEHEAD_1986}, \cite{Stallard_Todd_2005}).
Shrinkage, minimax, and Bayesian techniques have also been studied
\cite{Shrinkage}, \cite{minmax}.
For cases in which an unbiased estimator exists, 
the U-V estimator by Robbins can be used \cite{Robbins_1988}.
The current paper provides a general {\em{non-Bayesian}} framework,
 \ie where the parameters to be estimated after selection are deterministic
  and
 the underlying statistical models are general and admit general dependencies across parameters. 
In particular,  we establish the basic theory of $\Psi$-efficiency post selection estimation that includes the fundamental limits of estimation, ways to achieve efficiency when efficient estimator exists, and practical approaches.

In the context of signal processing,
the works in
\cite{hero2} and \cite{Hero3} investigate the Bayesian estimation after the detection of an unknown data region of interest.  The problem of post-detection estimation,
 or estimation after data censoring,
 is considered by  Chaumette, Larzabal, and Forster 
 \cite{Chaumette2005}, \cite{energy}, who derive a novel CRB on the conditional MSE,
involving conditional Fisher information.
It should be noted that in {\cite{Chaumette2005,energy,hero2,Hero3}}, the selection rule selects
the {\em{data}} to be used, while in our proposed model the {\em{parameter}} to be estimated is selected
and all the data can be used for estimation.
Selection and ranking  are highly related approaches \cite{mukhopadhyay1994multistage}.
Detection and estimation  after ranking and order statistics procedures
  are proposed in 
{\cite{Fishler_Messer_OS,Fishler_Messer2000}}
and are shown to have both practical and theoretical advantages  in terms of computational complexity and performance.
An empirical Bayes estimator for exponentially distributed populations is proposed in \cite{Efron_Tweedie}.
For the  problem of estimation after {\em{model}} selection, 
a bootstrap method for computing standard errors and confidence intervals  is considered
in \cite{Efron}, a post-selection  lasso method is developed in
\cite{Lee_Taylor_2014},
and
 the  CRB is derived in  \cite{Stoica} for model order selection.
However, it should be emphasized that in the case of estimation after parameter selection presented here, 
the measurement model is assumed to be  {\em{known}} and  there are  no modeling errors.
 In contrast, in estimation after  model selection
\cite{Efron}, \cite{Stoica}, the measurement model is {\em{unknown}} and is selected from a finite collection of competing models.

\subsection{Organization and notations}
 The remainder of the paper is organized as follows:
Section \ref{The_model} presents the mathematical model for the problem  of estimation after parameter selection.
The $\Psi$-unbiasedness in the  Lehmann sense and the $\Psi$-CRB  are derived in Section \ref{CRB_section} and estimation methods are developed
in Section \ref{estimation_methods_section},  for estimation after parameter selection.
Finally, the proposed $\Psi$-CRB and $\Psi$-unbiased estimators are   evaluated via simulations for the linear Gaussian model in Section \ref{simulation_section}.
Our conclusions appear in Section \ref{diss}.

In the rest of this paper, vectors are denoted by boldface
lowercase letters and matrices  by boldface uppercase
letters.
The operators  $(\cdot)^{\mbox{\tiny $T$}}$ and  $(\cdot)^{-1}$ denote the transpose and inverse,   respectively.
The vector $\evec_m \in{\mathbb{R}}^{M}$ is a vector of all zeros except for a 1 at the $m$th
position,  $\forall m=1,\ldots,M$, and
the $(m,k)$th element 
of the matrix $\Amat$ is denoted by $[ \Amat]_{m,k}$.
 The notations
 $\delta_{m,k}$
and ${\mathbf{1}}_A$ denote  the  Kronecker delta function and the indicator function of an
event $A$, respectively. 
The $m$th element of the gradient vector  $\nabla_\thetavecsmall  c$ 
is given by $\frac{\partial  c}{\partial \theta_m}$,
where $\thetavec=[\theta_1,\ldots,\theta_M]^{\mbox{\tiny $T$}}$, $c$ is an arbitrary scalar function of $\thetavec$,
$\nabla_\thetavecsmall^T  c \define (\nabla_\thetavecsmall  c)^T$, and
$\nabla_\thetavecsmall^2  c \define\nabla_\thetavecsmall \nabla_\thetavecsmall^T  c$.
The notations ${\rm{E}}_{\thetavecsmall} [\cdot]$ and ${\rm{E}}_{\thetavecsmall} [\cdot|A]$ represent the  expected and conditional-expected value
of its argument,
 parameterized by a deterministic parameter $\thetavec$
and given event $A$.

\section{Problem formulation}
\label{The_model}
Let $(\Omega_\xvec,{\cal{F}},P_\thetavecsmall)$ denote a probability space, where
 $\Omega_\xvec$ is the observation  space,
 ${\cal{F}}$ is the $\sigma$-algebra on $\Omega_\xvec$,
 and $\left\{P_\thetavecsmall \right\}$ is a family of probability measures parametrized by the deterministic parameter vector  $\thetavec=[\theta_1,\ldots,\theta_M]^{\mbox{\tiny $T$}}\in{\mathbb{R}}^M$.
Let $\hat{\thetavec}=[\hat{\theta}_1,\ldots,\hat{\theta}_M]^{\mbox{\tiny $T$}}$
 be an estimator of 
 $\thetavec$,   based on a random observation vector $\xvec\in\Omega_\xvec$,
\ie  $\hat{\thetavec} :\Omega_\xvec\rightarrow{\mathbb{R}}^M$.
For each probability measure  $P_\thetavecsmall$, the function $f(\xvec;\thetavec) $
 denotes the corresponding probability
 density function (pdf) of $\xvec$.
All the estimators in this paper are assumed to be in
the Hilbert space of absolutely
square integrable functions with respect to (w.r.t.) $P_\thetavecsmall$,
 ${\mathcal{L}}_2(\thetavec)$.

The basic structure of the proposed model for  estimation after parameter selection  consists of two stages: 
first, a parameter $\theta_m$ is selected 
according to a predetermined data-driven selection rule,
 $\Psi$,
 and then, this parameter is estimated.
	In this work, we assume that the selection rule  $\Psi$ is given 
and 
 we focus on the estimation  of the selected parameter.
The proposed model is presented schematically  in Fig. 
\ref{model2}. The extension for a selection of a subset of unknown parameters, \ie the multiparameter case,  is  discussed  in Section \ref{subset}.

A data-based {\em{selection rule}}
 is a deterministic function $\Psi:\Omega_\xvec\rightarrow\{1,\ldots,M\}$ that
selects a parameter based on the observation vector, $\xvec$.
That is, if $\Psi(\xvec)=m$, then the estimation goal is to estimate the parameter
$\theta_m$ based on the same observation vector $\xvec$.
We assume that the deterministic sets
 ${\mathcal{A}}_{m}\define \{\xvec: \xvec\in\Omega_\xvec,\Psi(\xvec)=m\}$,
$m=1,\ldots,M$,  partition $\Omega_\xvec$.
For the sake of simplicity of notation, in the following $\Psi(\xvec)$ is replaced by $\Psi$.
By using Bayes' rule it can be seen that the pdf of the observation conditioned on the event  $\xvec\in {\mathcal{A}}_{m}$
is
\be
\label{Bayes}
f(\xvec|\Psi=m;\thetavec)=
\frac{f(\xvec;\thetavec)}{\Pr(\Psi=m;\thetavec)},~~~
\forall \xvec\in {\mathcal{A}}_{m}
\ee
and is undefined otherwise,   where $ \Pr(\Psi=m;\thetavec)$ denotes the probability of  this event
for all $m=1,\ldots,M$.
\begin{figure}[htb]
\vspace{-0.25cm}
\centerline
{\psfig{figure=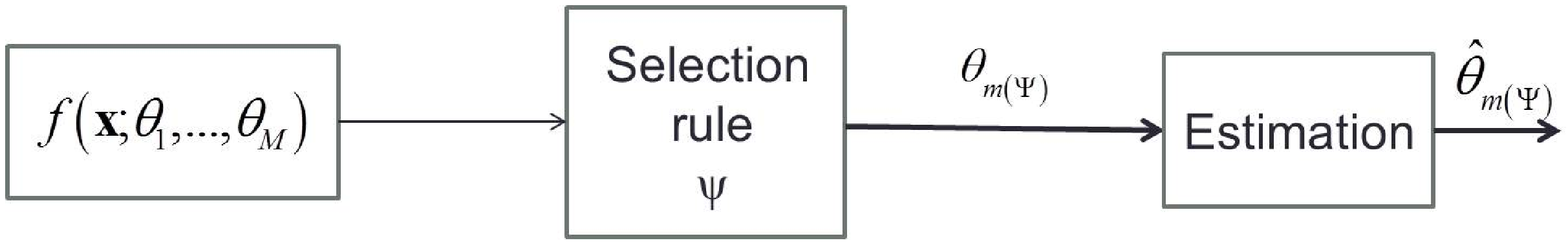,width=8cm}}
\caption{Schematic model  of  estimation after parameter selection.}
\label{model2}
\vspace{-0.25cm}
\end{figure}

A special case of the proposed model of estimation after parameter selection
is estimation in the presence of additional  undesired  deterministic nuisance parameters 
\cite{Kay_estimation}, \cite{Andrea_Mengali_Reggiannini_1994}.
Here, the selection of the desired parameter is performed in advance,
independently of the data, $\xvec$. Therefore, 
the  statistical characteristics, such as CRB and bias, are not affected by the selection process and are equal to those of the multiparameter estimation, in which the nuisance parameters are  estimated as well \cite{Kay_estimation}, \cite{Noam_Messer_2009}.

A more challenging and relevant  application of the proposed model is the classical estimation after  selection  with independent populations  \cite{Sarkadi}, \cite{Putter_Rubinstein},
which is presented schematically  in Fig. \ref{model1}. 
In  estimation after selection  with independent populations,
a given set of $M$ independent populations is assumed. These populations
might represent, for example, a set of $M$ different communication channels.
For any $m=1,\ldots,M$, it is supposed that  $N_m \geq 1$  random observations  are drawn from the $m$th population
to generate the  $m$th observation vector, $\yvec_m=[ y_m[0],\ldots,y_m[N_m-1]]^{\mbox{\tiny $T$}}$, with the associated marginal pdf, $f_m(\yvec_m;\theta_m)$,
in which  $\theta_m \in{\mathbb{R}}$  
denotes  the  unknown parameter related to  the $m$th population.
In this case, the  observation vector is given by
$\xvec=[\yvec_1^{\mbox{\tiny $T$}},\ldots,\yvec_M^{\mbox{\tiny $T$}}]^{\mbox{\tiny $T$}}$,
the joint pdf of the $M$ populations  is 
 $f(\xvec;\thetavec)=\prod_{m=1}^M f_m(\yvec_m;\theta_m) $, and  the selection  of a  parameter $\theta_m$ is equivalent 
to the selection of the $m$th  population or channel.
\vspace{-0.25cm}
\begin{figure}[htb]
\centerline
{\psfig{figure=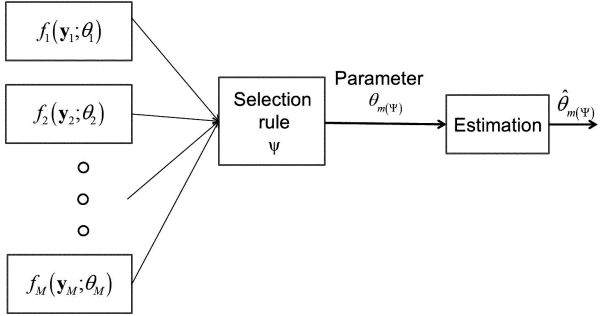,width=8cm}}
\caption{Schematic model  of  classical estimation after selection  with independent populations.}
\label{model1}
\vspace{-0.25cm}
\end{figure}

In this work, we are interested in the parameter estimation of the unknown deterministic vector 
$\thetavec$,  where only estimation errors  of the selected parameter are taken into consideration and the selection rule is  predetermined.
Therefore,  for a given selection rule, $\Psi$,
we use the following post-selection squared-error  (PSSE) cost {\cite{BIMJ200810442,mukhopadhyay1994multistage,Sarkadi,ICASSP2014_est_after_sel}}:
\beqna
\label{cost_def}
C^{(\Psi)}(\hat{\thetavec},\thetavec)\define\sum_{m=1}^M (\hat{\theta}_m-\theta_m)^2{\mathbf{1}}_{\{\Psi=m\}}.
\eeqna
The corresponding  PSMSE is given by 
\beqna
\label{PSMSE}
{\rm{E}}_{\thetavecsmall}\left[C^{(\Psi)}(\hat{\thetavec},\thetavec)\right]
={\rm{E}}_{\thetavecsmall}\left[\sum_{m=1}^M (\hat{\theta}_m-\theta_m)^2{\mathbf{1}}_{\{\Psi=m\}}\right]
\nonumber\\=
\sum_{m=1}^M \Pr(\Psi=m;\thetavec){\rm{E}}_{\thetavecsmall}\left[(\hat{\theta}_m-\theta_m)^2|\Psi=m\right],
\eeqna
where  (\ref{PSMSE}) 
is calculated by using  the densities $f(\xvec;\thetavec)$ and $f(\xvec|\Psi=m;\thetavec)$ for
the first and the second terms, respectively.
The  component-wise  PSMSE of a specific parameter is defined as
 \[ {\rm{E}}_{\thetavecsmall}\left[\left.(\hat{\theta}_m-\theta_m)^2\right|\Psi=m\right], ~~~m=1,\ldots,M.\]
The use of the indicator functions implies that the PSMSE may be equal for two  estimators 
that have different values with a  nonzero probability, \ie outside the subset indicated by these indicators. 
In fact, $\hat{\theta}_m$ affects the PSMSE
only for  observations $\xvec\in {\mathcal{A}}_{m}$.

In the mathematical statistics literature,  the unknown parameter for  estimation after  selection  with independent populations  is usually defined as 
$
\sum_{m=1}^M \theta_m {\mathbf{1}}_{\{\Psi=m\}}$, which  has both random and deterministic components.
In this work, we are interested in the estimation of the {\em{deterministic}} parameter $\thetavec$. 
 The notion of non-Bayesian estimation allows the derivation of the corresponding CRB-type lower bound
and  non-Bayesian estimation methods.

\section{The Cram$\acute{\text{e}}$r-Rao-type bound for estimation after parameter selection}
\label{CRB_section}
The CRB (e.g. \cite{Kay_estimation}, \cite{point_est})   provides a lower bound on the MSE of any mean-unbiased 
estimator and is used as a benchmark to study the optimality of practical parametric estimators.
In this section, a CRB-type lower bound for estimation after parameter selection is derived.
The proposed  bound is a lower bound on the PSMSE
of any  Lehmann unbiased estimator, as described in the following.
\subsection{$\Psi$-unbiasedness}
The mean-unbiasedness constraint is   commonly used in non-Bayesian parameter estimation \cite{Kay_estimation}.
However, a mean-unbiased estimator  is inappropriate
for  estimation after parameter selection problems, since we are interested only in errors of the selected parameter
 (See, e.g.  \cite{mukhopadhyay1994multistage}, \cite{Vellaisamy2009}).
Lehmann \cite{Lehmann}  proposed a generalization of the  unbiasedness concept, which is based on the
considered cost function. 
In this section, the general Lehmann unbiasedness is used to define the unbiasedness for
 estimation after parameter selection problems. 

\begin{definition}
\label{unbiased_definition}
The  estimator 
$\hat{\thetavec}$  is said to be an unbiased estimator of $\thetavec$ in the Lehmann sense \cite{Lehmann} w.r.t. the scalar nonnegative cost function $C(\hat{\thetavec},\thetavec)$ 
 if 
\be
\label{defdef}
{\rm{E}}_{\thetavecsmall}\left[C(\hat{\thetavec},\etavec) \right] \geq {\rm{E}}_{\thetavecsmall}\left[C(\hat{\thetavec},\thetavec) \right],~~~\forall \etavec,\thetavec\in\Omega_\thetavecsmall,
\ee
where  $\Omega_\thetavecsmall$ is the parameter space.
\end{definition}

The Lehmann unbiasedness definition implies that an estimator is unbiased  if, on average, 
 it is ``closer'' to  the true parameter $\thetavec$
  than to any other value in the parameter space.
The measure of ``closeness'' between the estimator  and the parameter is the  mean of the  cost function,
 $C(\hat{\thetavec},\thetavec)$. 
For example, 
it is shown in \cite{Lehmann} that under the squared-error cost function the 
Lehmann unbiasedness  in (\ref{defdef}) is reduced to the conventional mean-unbiasedness, 
${\rm{E}}_{\thetavecsmall}[\hat{\thetavec}]=\thetavec$, $\forall \thetavec\in \Omega_\thetavecsmall$.
Additional examples for  Lehmann unbiasedness with different cost functions can be found, for example,  in 
\cite{Lehmann}, \cite{PCRB_J}, and \cite{Routtenberg_cyclic}.
The following proposition describes the Lehmann unbiasedness for the
estimation after parameter selection,
named {\em{$\Psi$-unbiasedness}}.
\begin{proposition}
\label{unbiasedness_prop}
An estimator 
$\hat{\thetavec}:\Omega_\xvec\rightarrow {\mathbf{R}}^M$
is  an unbiased estimator of $\thetavec\in {\mathbf{R}}^M$ in the Lehmann sense  w.r.t. the
PSSE cost and the selection rule $\Psi$ {\em{iff}}
 \beqna
 \label{unbias6}
 {\rm{E}}_{\thetavecsmall}\left[(\hat{\theta}_m-\theta_m){\mathbf{1}}_{\{\Psi=m\}}
\right]=0,~\forall m=1,\ldots,M,~\forall \thetavec \in{\mathbb{R}}^M,
\eeqna
or, equivalently, 
\beqna
\label{10}
{\rm{E}}_{\thetavecsmall}\left[\hat{\theta}_m-\theta_m|\Psi=m\right]=0,
~~~\forall \thetavec\in{\mathbb{R}}^M
\eeqna
for all $ m=1,\ldots,M$ such that  $\Pr(\Psi=m;\thetavec )\neq 0$.
\end{proposition}
{\em{Proof:}} The proof appears  in Appendix A.

It can be seen that  the Lehmann unbiasedness definition in (\ref{unbias6}) and (\ref{10})  is a function of the given selection rule. 
Therefore, in the following,
 an estimator $\hat{\thetavec}$
 is said to be an $\Psi$-unbiased estimate of $\thetavec$ for the selection rule
 $\Psi$, if
 (\ref{unbias6})   (or, equivalently, (\ref{10}))   is satisfied.
The concept of risk-unbiased in the Lehmann sense for the classical estimation after selection
with independent populations
 has been discussed in the literature  for the random parameter $
\sum_{m=1}^M \theta_m {\mathbf{1}}_{\{\Psi=m\}}$ and 
various cost functions (e.g. \cite{conditionally1} and \cite{Nematollahi_Shariati_2012}).

\subsection{The $\Psi$-CRB}
Obtaining  an estimator with the minimum  PSMSE   among all $\Psi$-unbiased estimators  is
usually not tractable and a uniform $\Psi$-unbiased minimum PSMSE  estimator may not exist \cite{Sarkadi}.
 Thus,
lower bounds on the performance of any  $\Psi$-unbiased estimator are useful for performance
analysis and system design.
In the following, a new version of the CRB for estimation after parameter selection is derived.
It should be noticed  that, in general, the minimum PSMSE estimator is not unique 
since only  the estimation errors of the selected parameter are taken into consideration.

Let us define the following 
   {\em{post-selection Fisher information matrix (PSFIM)}} of the $m$th component:
\beqna
\label{JJJdef}
\Jmat_m(\thetavec,\Psi)&\define& {\rm{E}}_{\thetavecsmall}
\left[\nabla_\thetavecsmall \log f(\xvec|\Psi=m;\thetavec)
\right.
\nonumber\\&&
\left.\left.\times
\nabla_\thetavecsmall^T \log f(\xvec|\Psi=m;\thetavec)  \right|\Psi=m
\right],
\eeqna
for all $m=1,\ldots,M$.
In addition, we define the following  conditions that are a modified form of the well-known CRB regularity conditions 
  (e.g. \cite{point_est},  pp. 440-441).
\renewcommand{\theenumi}{C.\arabic{enumi}} 
\begin{enumerate}
\item
\label{cond1}
The post-selection likelihood gradient vector,
$\nabla_\thetavecsmall \log f(\xvec|\Psi=m;\thetavec)$,
 exists and is finite 
for any $\thetavec\in {\mathbb{R}}^M$,
 $\xvec\in{\mathcal{A}}_{m}$, and
 $\forall m=1,\ldots,M$.
That is, we assume that
the $m$th   PSFIM, $\Jmat_m(\thetavec,\Psi)$,
is a well-defined,  nonsingular, and nonzero matrix
for any $ \thetavec\in{\mathbb{R}}^M$ and $\forall m=1,\ldots,M$. 
\item
\label{cond2}  The 
operations of integration w.r.t. $\xvec$ and differentiation w.r.t. $\thetavec$ can be interchanged    as follows:
\beqna
\label{uu}
\int_{{\mathcal{A}}_{m}}\nabla_\thetavecsmall \left(g(\xvec,\thetavec)f(\xvec|\Psi=m;\thetavec)\right){\ud}\xvec\hspace{2cm}\nonumber\\=\nabla_\thetavecsmall
{\rm{E}}_{\thetavecsmall}\left[ g(\xvec,\thetavec)|\Psi=m\right],
\eeqna
for any $\thetavec\in{\mathbb{R}}^M$ and
for any differentiable and measurable  function   $g(\xvec,\thetavec)$.
\end{enumerate} 
\renewcommand{\theenumi}{\arabic{enumi}}
\begin{theorem}
\label{Th1} ($\Psi$-CRB)
Let the  regularity conditions \ref{cond1}-\ref{cond2}  be satisfied
and $\hat{\thetavec}$ be an $\Psi$-unbiased
estimator of $\thetavec\in{\mathbf{R}}^M$ for a given selection rule, $\Psi$, 	with a finite second moment.
Then,
the PSMSE is bounded by 
the following Cram$\acute{\text{e}}$r-Rao-type lower bound:
\beqna
\label{bound1}
{\rm{E}}_{\thetavecsmall}\left[C^{(\Psi)}(\hat{\thetavec},\thetavec)\right]
\geq  B^{(\Psi)}(\thetavec),
\eeqna
where 
\beqna
\label{bound2}
B^{(\Psi)}(\thetavec)&\define& \sum_{m=1}^M \Pr(\Psi=m;\thetavec)\left[\Jmat_m^{-1}(\thetavec,\Psi)\right]_{m,m},
\eeqna
and  $\Jmat_m(\thetavec,\Psi)$ is the PSFIM defined in 
(\ref{JJJdef}).
Furthermore, the component-wise $\Psi$-CRB  on the PSMSE of a specific parameter  is given by
\be
\label{bound_marginal}
{\rm{E}}_{\thetavecsmall}\left[(\hat{\theta}_m-\theta_m)^2|\Psi=m\right]
\geq
\left[\Jmat_m^{-1}(\thetavec,\Psi)\right]_{m,m},
\ee
for all $ m=1,\ldots,M$.
The equality holds in (\ref{bound1}) and (\ref{bound_marginal})  {\em{iff}} there exist   functions
 $h_m(\thetavec)$, $m=1,\ldots,M$, such that
\beqna
\label{eq_cond}
\sum_{l=1}^M 
\frac{\partial \log f(\xvec|\Psi=m;\thetavec)}{\partial \theta_l}\left[\Jmat_m^{-1}(\thetavec,\Psi)\right]_{l,m}=\hspace{1cm}
\nonumber\\h_m(\thetavec)
 (\hat{\theta}_m-\theta_m),~~~ \forall m=1,\ldots,M
\eeqna
 almost surely (a.s.) $\xvec\in{\mathcal{A}}_{m}$.
\end{theorem}
\begin{proof}
 According to the Cauchy-Schwarz inequality:
\be
\label{CS}
{\rm{E}}_{\thetavecsmall}
\left[\left.\eta^2(\xvec,\thetavec)\right|\Psi=m\right]
\geq  \frac{{\rm{E}}_{\thetavecsmall}^2\left[\left.\eta(\xvec,\thetavec)d(\xvec,\thetavec)
\right|\Psi=m\right]}
{{\rm{E}}_{\thetavecsmall}\left[\left.d^2(\xvec,\thetavec)\right|\Psi=m\right]},
\ee
for any  measurable  functions $\eta(\xvec,\thetavec)$ and $h(\xvec,\thetavec)$  with finite second moments.
By substituting $\eta(\xvec,\thetavec)=\hat{\theta}_m-\theta_m$ and
 \[d(\xvec,\thetavec)=\sum_{l=1}^M \frac{\partial \log f(\xvec|\Psi=m;\thetavec)}{\partial \theta_l}\left[\Jmat_m^{-1}(\thetavec,\Psi)\right]_{l,m}\] in (\ref{CS})
 and under Condition \ref{cond1},  one obtains
 \beqna
 \label{CS2}
 {\rm{E}}_{\thetavecsmall}
\left[\left.(\hat{\theta}_m-\theta_m)^2\right|\Psi=m\right]
\geq \hspace{3.75cm}\nonumber\\ \frac{  
 \left(\sum\limits_{l=1}^M a_{l,m}(\thetavec)
\left[\Jmat_m^{-1}(\thetavec,\Psi)\right]_{l,m}
\right)^2}
{\sum\limits_{l=1}^M\sum\limits_{k=1}^M  \left[\Jmat_m(\thetavec,\Psi)\right]_{k,l}
\left[\Jmat_m^{-1}(\thetavec,\Psi)\right]_{l,m}
\left[\Jmat_m^{-1}(\thetavec,\Psi)\right]_{m,k}},
\eeqna
for any estimator with ${\rm{E}}_{\thetavecsmall}\left[\left.(\hat{\theta}_m-\theta_m)^2\right|\Psi=m\right]<\infty$
and a nonsingular PSFIM,
where
\[a_{l,m}(\thetavec)\define
{\rm{E}}_{\thetavecsmall}\left[\left.\frac{\partial \log f(\xvec|\Psi=m;\thetavec)}{\partial \theta_l}
(\hat{\theta}_m-\theta_m)\right|\Psi=m\right],
\]
for all $ m,l=1,\ldots,M$.
According to the Cauchy-Schwarz conditions, the equality in (\ref{CS2})
holds {\em{iff}} (\ref{eq_cond}) is satisfied for $m\in\{1,\ldots,M\}$.
By using
integration by parts and assuming   Condition \ref{cond2}, it can be verified that
\beqna
\label{by_parts}
a_{l,m}(\thetavec)=
\frac{\partial}{\partial \theta_l}{\rm{E}}_{\thetavecsmall}\left[\left.(\hat{\theta}_m-\theta_{m})
\right|\Psi=m\right]
+\delta_{l,m}
=\delta_{l,m},
\eeqna
for all $m,l=1,\ldots,M$,
where the last equality is obtained 
by using the   $\Psi$-unbiasedness  conditions from (\ref{10}).
In addition, it can be verified that
\beqna
\label{Jthree}
{\sum\limits_{l=1}^M\sum\limits_{k=1}^M  \left[\Jmat_m(\thetavec,\Psi)\right]_{k,l}
\left[\Jmat_m^{-1}(\thetavec,\Psi)\right]_{l,m}
\left[\Jmat_m^{-1}(\thetavec,\Psi)\right]_{m,k}}
\nonumber\\
=\left[\Jmat_m^{-1}(\thetavec,\Psi)\right]_{m,m}.\hspace{3cm}
\eeqna
By substituting    (\ref{by_parts}) and (\ref{Jthree})
 in (\ref{CS2}), we obtain
  the component-wise $\Psi$-CRBs on the PSMSE in (\ref{bound_marginal}).
Then, by  multiplying  (\ref{bound_marginal}) by $\Pr(\Psi=m;\thetavec)$ and  taking the sum of over   $m=1,\ldots,M$, we obtain 
 the $\Psi$-CRB  in (\ref{bound1}).
Furthermore, the equality condition in (\ref{eq_cond}) stems from the equality conditions of (\ref{CS2}).
\end{proof}

The following Lemma presents two alternative formulations of the PSFIM that are based on
the  (unconditional) likelihood function and the probability of selection, instead of the conditional likelihood used in (\ref{JJJdef}).
These formulations can be more tractable for further estimation and sampling procedures.
\begin{lemma}
\label{lemma1}
Assuming that Conditions \ref{cond1}-\ref{cond2} are satisfied,
 the second derivative w.r.t.  $\thetavec$ of
 $f(\xvec|\Psi=m ;\thetavec)$ exists and is bounded and continuous  
$\forall \xvec\in{\mathcal{A}}_{m}$,
and
the integral
 $\int_{{\mathcal{A}}_{m}} f(\xvec|\Psi=m ;\thetavec){\ud}\xvec$
 is twice differentiable under the integral
sign for all $ m=1,\ldots,M$, $ \thetavec\in {\mathbb{R}}^M$. 
 Then, the $m$th  PSFIM in (\ref{JJJdef})
satisfies
\beqna
\label{JJJdef1.5}
\Jmat_m(\thetavec,\Psi)=  {\rm{E}}_{\thetavecsmall}\left[\left.
\nabla_\thetavecsmall  \log f(\xvec;\thetavec)\nabla_\thetavecsmall^T
 \log f(\xvec;\thetavec) \right|\Psi=m
\right]
\nonumber\\-
\nabla_\thetavecsmall \log \Pr(\Psi=m;\thetavec)
\nabla_\thetavecsmall^T \log \Pr(\Psi=m;\thetavec)
\eeqna
and
\beqna
\label{JJJdef2}
\Jmat_m(\thetavec,\Psi)=
 &-&{\rm{E}}_{\thetavecsmall}\left[\left.
\nabla_\thetavecsmall^2 \log f(\xvec;\thetavec) \right|\Psi=m\right]
\nonumber\\
&+&\nabla_\thetavecsmall^2 \log\Pr(\Psi=m;\thetavec)
,
\eeqna
for all $ m=1,\ldots,M$, $ \thetavec\in {\mathbb{R}}^M$, and for any selection rule $\Psi$.
\end{lemma}
{\em{Proof:}} The proof appears  in Appendix B.

Similar to the  $\Psi$-CRB from Theorem \ref{Th1}, 
by using the Cauchy-Schwarz inequality and the $\Psi$-unbiasedness
we can obtain various non-Bayesian bounds on the PSMSE 
These bounds  are
modifications of various non-Bayesian bounds, such as the 
biased CRB and  Barankin-type bounds {\cite{Barankin,Chaumette_Larzabal_2008,TTB}}.


\subsection{Special cases}
\label{examples_sec}
In this section, we demonstrate the proposed $\Psi$-CRB and $\Psi$-unbiasedness for
  different cases.
\subsubsection{Randomized  selection rule}
  The  randomized, coin-flipping selection rule satisfies
$\Pr(\Psi_{\text{rand}}=m;\thetavec)=p_m$,
   for all $m=1,\ldots,M$,
	where $\{p_m\}\in[0,1]^M$
 are  independent of $\xvec$.
Therefore, the $\Psi$-unbiasedness from (\ref{unbias6}) in this case is given by
 \beqna
 \label{unbias6_flip}
{\rm{E}}_{\thetavecsmall}\left[\hat{\theta}_{m}-\theta_m\right]=0,~
~~\forall \thetavec \in{\mathbb{R}}^M,
\eeqna
for all $ m=1,\ldots,M$ with $p_m\neq 0$.
The $\Psi$-unbiasedness in  (\ref{unbias6_flip})
is the classical mean-unbiasedness definition.
In this case, 
the $\Psi$-CRB from (\ref{bound2})  is reduced to 
\beqna
\label{bound1_flip}
B^{(\Psi_{\text{rand}})}(\thetavec)= \sum_{m=1}^M p_m \left[\Bmat(\thetavec)\right]_{m,m},
\eeqna
where $\Bmat(\thetavec)=\Jmat^{-1}\left(\thetavec\right)$
is the conventional CRB.
Thus, 
the proposed $\Psi$-CRB for the randomized selection rule, $\Psi_{\text{rand}}$,
is  equal to a  weighted sum of the diagonal elements of the classical CRB, $\Bmat(\thetavec)$,    for
estimating  $\thetavec$ 
 without a selection stage.
In particular, 
for 
$p_{m}=\delta_{m,m'}$,
 where $\theta_{m'}$ is the desired parameter,
we obtained an  estimation problem  in the presence of  nuisance parameters, \ie
where the selection of the ``parameter of interest" $\theta_{m'}$ is performed in advance.
It is easy to verify that in this case
the $\Psi$-CRB and $\Psi$-unbiasedness are reduced to  their classical, marginal versions.
This result coincide with the literature on non-Bayesian nuisance parameter estimation  (e.g. \cite{Kay_estimation} and \cite{Noam_Messer_2009}).
\subsubsection{Parameter coupling}
For conventional  parameter estimation with a diagonal FIM,
where the FIM is defined as 
\beqna
\Jmat(\thetavec)\define
{\rm{E}}_{\thetavecsmall}\left[
\nabla_\thetavecsmall\log f(\xvec;\thetavec)
\nabla_\thetavecsmall^T \log f(\xvec;\thetavec)
\right],
\eeqna
the unknown parameters
are  decoupled from each other; that is, knowledge of one parameter does
not affect the accuracy in the estimation of the others. 
This situation  occurs, for example, for
classical estimation after selection with independent populations, in which 
$f(\xvec;\thetavec)=\prod_{m=1}^M f_m(\yvec_m;\theta_m)$. 
 However, it should be noted that 
the PSFIMs are not necessarily  diagonal for diagonal FIM cases,
since the selection step may create  dependency and coupling between the parameters over the different populations.
For example, by using the form of the PSFIM in (\ref{JJJdef2}), 
it can be seen   that 
the matrix
$\nabla_\thetavecsmall^2 \log\Pr(\Psi=m;\thetavec)$
may be  a nondiagonal matrix for a data-dependent selection rule.
\subsubsection{Biased $\Psi$-CRB}
Similar to the proof of Theorem \ref{Th1},
it can be shown that under regularity conditions \ref{cond1}-\ref{cond2}
the PSMSE is bounded by the following biased $\Psi$-CRB:
\beqna
\label{bound1_biased}
{\rm{E}}_{\thetavecsmall}\left[C^{(\Psi)}(\hat{\thetavec},\thetavec)\right]
\geq   \sum_{m=1}^M \Pr(\Psi=m;\thetavec)\hspace{2.5cm}
\nonumber\\
\times\left(\nabla_\thetavecsmall b_m(\thetavec)+\evec_m \right)^T\left[\Jmat_m^{-1}(\thetavec,\Psi)\right]\left(\nabla_\thetavecsmall b_m(\thetavec)+\evec_m \right),
\eeqna
for any   $\Psi$-biased estimator, $\hat{\thetavec}$, with the biases
\[
b_m(\thetavec)={\rm{E}}_{\thetavecsmall}\left[\hat{\theta}_m-\theta_m|\Psi=m\right],~m=1,\ldots,M,
\]
and a finite second moment.

\subsection{Estimation after parameter subset selection}
\label{subset}
In many problems,  we are interested
in selecting a {\em{subset}} of parameters 
 and  then, estimating  the parameters of the selected subset  {\cite{Jeyaratnam_Panchapakesan_1982,Gupta_1965}}.
This subset may be of random size,
with/without overlapping between the subspaces.
The selection rule is  $\Psi:\Omega_\xvec\rightarrow \{\Psi_1,\ldots,\Psi_L\}$,
where $\{\Psi_1,\ldots,\Psi_L\}$ is a finite covering  of the set $\{1,\ldots,M\}$, \ie it is a
division of  $\{1,\ldots,M\}$
  as a union of possibly-overlapping  non-empty $L$ subsets, such as the power set.
In this case, the PSSE cost function from (\ref{cost_def}) is replaced by
\[
C^{(\Psi)}(\hat{\thetavec},\thetavec)\define\sum_{m=1}^M (\hat{\theta}_m-\theta_m)^2{\mathbf{1}}_{\{m\in\Psi\}}
\]
and the corresponding  PSMSE  is: 
\beqna
{\rm{E}}_{\thetavecsmall}\left[C^{(\Psi)}(\hat{\thetavec},\thetavec)\right]
&=&\sum_{l=1}^L \Pr(\Psi=\Psi_l;\thetavec)
\nonumber\\&&\times 
\sum_{m=1,m\in\Psi_l}^M  {\rm{E}}_{\thetavecsmall}\left[(\hat{\theta}_m-\theta_m)^2|\Psi=\Psi_l\right].\nonumber
\eeqna
Similar to Proposition \ref{unbiasedness_prop} and Theorem \ref{Th1}, the $\Psi$-unbiasedness 
and $\Psi$-CRB for subset selection are,  respectively, given by
 \beqna
 \label{unbias6_s}
 {\rm{E}}_{\thetavecsmall}\left[\left.\hat{\theta}_m-\theta_m\right|\Psi=\Psi_l
\right]=0,~ \begin{array}{l}\forall m=1,\ldots,M,~m\in\Psi_l\\
\forall l=l=1,\ldots,L
\end{array}
\eeqna
for any $ \thetavec \in{\mathbb{R}}^M$
and
\[
{\rm{E}}_{\thetavecsmall}\left[C^{(\Psi)}(\hat{\thetavec},\thetavec)\right] \geq B^{(\Psi)}(\thetavec),
\]
where
\beqna
\label{B_subset}
B^{(\Psi)}(\thetavec)\define   \sum_{l=1}^L \Pr(\Psi=\Psi_l;\thetavec)\sum_{m=1,m\in\Psi_l}^M
\left[\Jmat_l^{-1}(\thetavec,\Psi)\right]_{m,m}
\eeqna
and
\beqna
\Jmat_l(\thetavec,\Psi)&\define& {\rm{E}}_{\thetavecsmall}\left[
\nabla_\thetavecsmall \log f(\xvec|\Psi=\Psi_l;\thetavec)\right.
\nonumber\\&&
\left.\left.\times
\nabla_\thetavecsmall^T\log f(\xvec|\Psi=\Psi_l;\thetavec) \right| \Psi=\Psi_l
\right].\nonumber
\eeqna
for all $l=1,\ldots,L$  is the PSFIM for this case.

If  the selection rule selects all the $M$ parameters, 
then, the PSMSE is equal to  the MSE and we obtain the conventional parameter estimation problem, mean-unbiasedness, 
and the well known CRB.
Another special case of estimation after parameter subset selection is the classical estimation after  selection model with independent populations, where
the pdf of  each single population is 
a function of
 {\em{multiple}} unknown  parameters. 
 For this nonoverlapped case, we can also obtain a matrix-form of the $\Psi$-CRB from (\ref{B_subset}) by using
the matrix form of  the Cauchy-Schwarz inequality and the vector $\Psi$-unbiasedness from (\ref{unbias6_s}).

\subsection{Estimation after data censoring}
A related problem is the estimation after {\em{data}} censoring,
which is  obtained from the aforementioned 
 model
by assuming
a selection rule
 that 
restricts
the set of observations available for parameter estimation. 
In this case,
 we use the observations only if $\Psi=1$, and  we remove them otherwise.
Similar to the derivation of the   $\Psi$-CRB  in Theorem \ref{Th1},
 the following matrix-form $\Psi$-CRB is obtained for this case:
\beqna
\label{bound_detection}
{\rm{E}}_{\thetavecsmall}\left[(\hat{\thetavec}-\thetavec)
(\hat{\thetavec}-\thetavec)^T|\Psi=1\right]
\geq
\Jmat_{c}^{-1}(\thetavec,\Psi),
\eeqna
where
\beqna
\Jmat_c(\thetavec,\Psi)&\define& {\rm{E}}_{\thetavecsmall}\left[\left.
\nabla_\thetavecsmall \log f(\xvec|\Psi=1;\thetavec)\right.\right.
\nonumber\\&&
\left.\left.\times
\nabla_\thetavecsmall^T \log f(\xvec|\Psi=1;\thetavec) \right|\Psi=1
\right].
\eeqna
The bound in  (\ref{bound_detection})  coincides with the conditional CRB derived in \cite{Chaumette2005}
for estimation after binary detection, \ie when a binary detection step is performed
 before the estimation of the parameters. 
This observation remains valid 
for any non-Bayesian bound on the PSMSE, which can be derived 
by using the Cauchy-Schwarz inequality and the $\Psi$-unbiasedness,
in a similar way to the  derivations in \cite{Chaumette2005}, \cite{TTB}.
However, it should be noted that in estimation after data censoring, the selection rule selects
the {\em{data}},  while in our model the {\em{parameter}} to be estimated has been selected.

\section{Post-selection estimation}
\label{estimation_methods_section}
\subsection{The PSML estimator}
For general parameter estimation, the commonly used ML estimator is defined as
\beqna
\label{ML}
\hat{\thetavec}^{({\mbox{\tiny ML}})}=\arg\max_\thetavecsmall\log f(\xvec;\thetavec).
\eeqna
It is well known that the ML estimator is inappropriate for estimation after parameter selection
since it  does not take into account the 
  prescreening process  \cite{Sarkadi}, \cite{Putter_Rubinstein}.
Inspired by Theorem \ref{Th1}, we define the PSML estimator
as:
\beqna
\label{PSML}
\hat{\thetavec}^{({\mbox{\tiny PSML}})}&=&\arg\max_\thetavecsmall\left\{\sum_{m=1}^M \log f(\xvec|\Psi=m;\thetavec){\mathbf{1}}_{\{\xvec\in {\mathcal{A}}_{m}\}}\right\}
\nonumber\\
&=&\arg\max_\thetavecsmall\left\{\log
f(\xvec;\thetavec)\right.
\nonumber\\
&&
-\left.\sum_{m=1}^M\log\Pr(\Psi=m;\thetavec){\mathbf{1}}_{\{\xvec\in {\mathcal{A}}_{m}\}}\right\}
,
\eeqna
where the last equality is obtained by using
(\ref{Bayes}).
	We propose using the PSML estimator instead of the ML estimator for
	estimation after parameter selection problems.
The PSML estimator can be interpreted as the ``penalized ML estimator" \cite{point_est}, where  the penalty term in this case is 
$-\sum_{m=1}^M\log\Pr(\Psi=m;\thetavec){\mathbf{1}}_{\{\xvec\in {\mathcal{A}}_{m}\}}$. 
However, since the penalty term is not  a probability density w.r.t. $\thetavec$,  (\ref{PSML}) does not have  a Bayesian interpretation.  
It can be seen that if the selection probability, $\Pr(\Psi=m;\thetavec)$, is not a function  of 
$\thetavec$, then the PSML estimator coincides with the ML estimator.
This situation occurs, for example, for a randomized selection rule and for estimation in the presence of nuisance parameters.
Under suitable regularity conditions, such as differentiability,  the  PSML
estimator is a solution to the following score equation
\be
\label{PSML2}
\sum_{m=1}^M \nabla_\thetavecsmall \log f(\xvec|\Psi=m;\thetavec)
{\mathbf{1}}_{\{\xvec\in {\mathcal{A}}_{m}\}}=\zerovec.
\ee

The $\Psi$-efficient estimator is defined as follows. 
\begin{definition}
\label{efficient_definition}
An  estimator 
  is said to be an $\Psi$-efficient estimator of $\thetavec$ 
 if it
is  an  $\Psi$-unbiased estimator that achieves the $\Psi$-CRB.
\end{definition}

It should be noticed that  
the requirement for equality condition in (\ref{eq_cond}), \ie for the
$\Psi$-CRB achievability,
 is 
relevant only 
in the subspace ${\mathcal{A}}_{m}$
and
the  estimation errors outside this region can have arbitrary values.
Thus, the estimator which satisfies the equality condition in (\ref{eq_cond}), if exists, is not unique, since    by changing
this  estimator outside the subset $ {\mathcal{A}}_{m}$
we obtain a new estimator  that attained  (\ref{eq_cond}).
In particular, the $\Psi$-efficient estimator is not unique.
The following theorem  describes the relations between the PSML and the $\Psi$-efficient estimators.
\begin{theorem}
\label{Th3}
Assume that
the  regularity conditions \ref{cond1}-\ref{cond2}  are satisfied and that
  $\hat{\thetavec}^{(\Psi{\text{-eff}})}$ is   an    $\Psi$-efficient estimator, as defined in Definition
	\ref{efficient_definition}.
Then, 
\be
\label{eq_cond5}
\hat{\theta}_m^{(\Psi{\text{-eff}})}=\hat{\theta}_m^{({\mbox{\tiny PSML}})},
~~~\forall m=1,\ldots,M,   ~\forall\xvec\in{\mathcal{A}}_{m}~a.s..
\ee 
\end{theorem}
\begin{proof}
According to Definition 	\ref{efficient_definition},
$\hat{\thetavec}^{(\Psi{\text{-eff}})}$
is  an  $\Psi$-unbiased estimator that achieves the $\Psi$-CRB.
According to (\ref{eq_cond}), the estimator that achieves the $\Psi$-CRB satisfies
\beqna
\label{eq_cond_estimator}
\hat{\theta}_m^{(\Psi{\text{-eff}})}=\hspace{6.5cm}\nonumber\\
\theta_m+
\frac{1}{h_m(\thetavec)}\sum_{l=1}^M \frac{\partial \log f(\xvec|\Psi=m;\thetavec)}{\partial \theta_l}
\frac{\left[\Jmat_m^{-1}(\thetavec,\Psi)\right]_{l,m}}{h_m(\thetavec)},
\eeqna
a.s. $\forall\xvec\in{\mathcal{A}}_{m}$,
for all $ m=1,\ldots,M$ and $ \thetavec\in \Omega_\thetavecsmall$.
By using  (\ref{PSML2}), it can be concluded  that
\be
\label{28}
\left.\frac{\partial \log f(\xvec|\Psi=m;\thetavec)}{\partial\theta_l}\right|_{\thetavecsmall=\hat{\thetavecsmall}^{({\mbox{\tiny PSML}})}}=0,~\forall l=1,\ldots,M,
\ee for $\xvec\in{\mathcal{A}}_{m}$.
Therefore, 
by substituting $\thetavec=\hat{\thetavec}^{({\mbox{\tiny PSML}})}$  and (\ref{28}) in (\ref{eq_cond_estimator}),
one obtains 
\[
\hat{\theta}_m^{(\Psi{\text{-eff}})}=
\hat{\theta}_m^{({\mbox{\tiny PSML}})},~~~\forall m=1,\ldots,M,~\xvec\in{\mathcal{A}}_{m}.
\]
 Thus,
(\ref{eq_cond5}) is satisfied.
 \end{proof}
It should be noted that (\ref{eq_cond5}) implies that 
for any observation vector $\xvec\in{\mathcal{A}}_{m}$, the $m$th elements of the $\Psi$-efficient estimator  and the PSML estimator are identical, while the other elements  may be different.
However, these elements have no influence on the PSMSE.

\subsection{Practical implementations of the PSML}
\label{itr_section}
In practice,  an analytical expression for the PSML in (\ref{PSML2}) is usually unavailable due to the intractability of the probability of selection.
 In the following, we propose three iterative methods  for the implementation of the PSML:
  1) Newton-Raphson,  2) post-selection Fisher scoring, and
	3) maximization by parts (MBP).
	These methods are based on the assumptions that the post-selection likelihood, $f(\xvec|\Psi=m;\thetavec)$, is a twice continuously differentiable function  w.r.t. $\thetavec$ for any $m=1,\ldots,M$, and that 
	 a unique solution to the score equation in 
	(\ref{PSML2}) exists, which is the PSML estimator, $\hat{\thetavec}^{({\mbox{\tiny PSML}})}$.

\subsubsection{Newton-Raphson}
The  Newton-Raphson method for solving the  post-selection likelihood equation in (\ref{PSML2})
is based on replacing the objective function on the r.h.s. of (\ref{PSML})
 by its first-order Taylor expansion (see, e.g. Chapter 7 of  \cite{Kay_estimation}).
Therefore, 
 the $i$th iteration of the Newton-Raphson method is given by:
\beqna
\label{iterative_NR}
\hat{\thetavec}^{(i+1)}=\hat{\thetavec}^{(i)}-\sum_{m=1}^M\Hmat_m^{-1}\left(\hat{\thetavec}^{(i)},\Psi\right)\times\hspace{2.4cm}\nonumber\\
\left. \nabla_\thetavecsmall \log f(\xvec|\Psi=m;\thetavec)
\right|_{\thetavecsmall=\hat{\thetavecsmall}^{(i)}}{\mathbf{1}}_{\{\xvec\in {\mathcal{A}}_{m}\}},
\eeqna
$\forall i=1,2,\ldots$, where the  Hessian matrix is given by
\beqna
\label{Hmat}
\Hmat_m(\thetavec,\Psi)\define
\nabla_\thetavecsmall^2\log f(\xvec|\Psi=m;\thetavec)
,~~~\forall \thetavec\in{\mathbb{R}}.
\eeqna

\subsubsection{Post-selection Fisher scoring}
Similar to the derivation of Fisher scoring for ML estimation  \cite{Kay_estimation},
a variation on (\ref{iterative_NR})  is the Fisher scoring method in which the post-selection Hessian is replaced  by its expected value, $-\Jmat_m(\thetavec,\Psi)$.
Thus, 
the $i$th iteration of the resulting
post-selection Fisher scoring
procedure  is given by:
\beqna
\label{fisher_scoring}
\hat{\thetavec}^{(i+1)}=\hat{\thetavec}^{(i)}+\sum_{m=1}^M
\Jmat_m^{-1}\left(\hat{\thetavec}^{(i)},\Psi \right)\times\hspace{2.4cm}\nonumber\\
\left.\nabla_\thetavecsmall \log f(\xvec|\Psi=m;\thetavec)
\right|_{\thetavecsmall=\hat{\thetavecsmall}^{(i)}}{\mathbf{1}}_{\{\xvec\in {\mathcal{A}}_{m}\}},
\eeqna
for $ i=1,2,\ldots$, where the PSFIM  is defined in (\ref{JJJdef}).

\subsubsection{MBP}
\label{max_by_part}
In some instances, the second derivative of the  post-selection likelihood function is  intractable, so that calculation of $\Hmat_m^{-1}\left(\thetavec,\Psi \right)$ and $\Jmat_m^{-1}\left(\thetavec,\Psi \right)$ 
 may be difficult. Thus, the Newton-Raphson and  post-selection Fisher scoring methods are intractable.
In \cite{Song_Fan_Kalbfleisch_2005}, an MBP algorithm is proposed   that strategically selects a part of the full likelihood function with easily computed second-order derivatives.  The remaining more difficult part of the likelihood function participates in the algorithm in such a way that its second-order derivative is not needed. 
If the “information dominance condition”  \cite{Song_Fan_Kalbfleisch_2005} is satisfied, then the MBP estimator converges to the MSPL estimator and its asymptotic performance is better than that  of the ML estimator.

In the context of  the estimation after parameter selection model, according to the r.h.s. of (\ref{PSML}), 
 the PSML consists of maximizing the sum of two functions:
$\log
f(\xvec;\thetavec)$ and
$-\sum_{m=1}^M\log\Pr(\Psi=m;\thetavec){\mathbf{1}}_{\{\xvec\in {\mathcal{A}}_{m}\}}$.
Maximizing the  probability of selection is usually  less tractable than minimizing the likelihood function and
may create  dependency and coupling between the different parameters.
Thus, we use the MBP method \cite{Song_Fan_Kalbfleisch_2005}, such that
the $i$th iteration   is given by
\beqna
\label{by_parts_max}
&&\left.\nabla_\thetavecsmall \log f(\xvec;\thetavec)\right|_{\thetavecsmall=\hat{\thetavecsmall}^{(i+1)}}=\nonumber\\&&
\sum_{m=1}^M\left. \nabla_\thetavecsmall \log\Pr(\Psi=m;\thetavec)\right|_{\thetavecsmall=\hat{\thetavecsmall}^{(i)}}{\mathbf{1}}_{\{\xvec\in {\mathcal{A}}_{m}\}}
\eeqna
and the initial estimate is the ML estimator, \ie $\hat{\thetavec}^{(0)}=\hat{\thetavec}^{({\mbox{\tiny ML}})}$.
The advantage of 
the estimation iteration in (\ref{by_parts_max}) is that there is no need for  a second derivative of the probability of selection.

Asymptotically as $N\rightarrow \infty$, the MBP iteration converges
to the PSML estimator
under the   ``information dominance condition"  \cite{Song_Fan_Kalbfleisch_2005}, \ie
if 
\[\left|\left| \Jmat^{-1}(\thetavec) \nabla_\thetavecsmall  \log \Pr(\Psi=m;\thetavec)
\nabla_\thetavecsmall^T \log \Pr(\Psi=m;\thetavec) \right|\right|<1,
\]
and  the Fisher information 
is larger than the information contained in the probability of selection,
where  $||\cdot||$ denotes the  spectral norm.

Additional relaxation can be achieved 
by  using the  Newton-Raphson method on the l.h.s. of
(\ref{by_parts_max}), \ie
by \cite{Liao_Qaqish_2005}:
\beqna
\label{iterative_NR2}
\hat{\thetavec}^{(i+1)}=\hat{\thetavec}^{(i)}-
\Hmat^{-1}\left(\hat{\thetavec}^{(i)}\right)\times\hspace{3cm}\nonumber\\\sum_{m=1}^M 
\left. \nabla_\thetavecsmall \log f(\xvec|\Psi=m;\thetavec)
\right|_{\thetavecsmall=\hat{\thetavecsmall}^{(i)}}{\mathbf{1}}_{\{\xvec\in {\mathcal{A}}_{m}\}},
\eeqna
or by using the Fisher scoring variation:
\beqna
\label{fisher_scoring2}
\hat{\thetavec}^{(i+1)}=\hat{\thetavec}^{(i)}+
\Jmat^{-1}\left(\hat{\thetavec}^{(i)} \right)\times\hspace{3cm}\nonumber\\\sum_{m=1}^M 
\left. \nabla_\thetavecsmall \log f(\xvec|\Psi=m;\thetavec)
\right|_{\thetavecsmall=\hat{\thetavecsmall}^{(i)}}{\mathbf{1}}_{\{\xvec\in {\mathcal{A}}_{m}\}},
\eeqna
 where the  Hessian matrix  in this case is given by
$
\Hmat(\thetavec)\define
\nabla_\thetavecsmall^2 \log f(\xvec;\thetavec).
$
One merit of the iterative methods in (\ref{iterative_NR2}) and
(\ref{fisher_scoring2}) is that the MBP  utilizes the conventional Hessian and FIM to direct the search for the PSML.
This is very useful for the classical estimation after selection problem  with independent populations, where the
conventional Hessian and Fisher scoring  are diagonal matrices.
In Appendix C, an iterative method is proposed for cases with intractable probability of selection.



\section{Examples}
\label{simulation_section}
\subsection{Uniform distribution}
Consider the following observation model:
\beqna
\label{model_uniform}
y_m[n] \sim U[0,\theta_m] ,~n=0,\ldots,N-1,~m=1,2,
\eeqna
where $U[a,b]$ denotes the continuous uniform distribution on the support $[a,b]$
and  the two populations are assumed to be independent.
For the selection of the population with the largest maximum,
the 
 SMS rule, $\Psi_{{\text{SMS}}}$, selects the population with the largest sufficient statistics,   \ie 
\be
\label{SMS_uniform}
\Psi_{{\text{SMS}}}=\arg \max_{m=1,2}\left\{\hat{\theta}_m^{({\mbox{\tiny ML}})}\right\},
\ee
where the ML estimator of $\theta_m$ is given by
\be
\label{ML_est_uniform}
\hat{\theta}_m^{({\mbox{\tiny ML}})}\define \max_{n=0,\ldots,N-1}\left\{ y_{m}[n]\right\},~~~m=1,2.
\ee
The uniform minimum variance unbiased (MVU)  estimator  (in the conventional sense)
for this problem and without a selection stage,
is given by (e.g. \cite{Kay_estimation}, p. 115)
\be
\label{MVU_est_uniform}
\hat{\theta}_m^{({\mbox{\tiny MVU}})}\define \frac{N+1}{N}\hat{\theta}_m^{({\mbox{\tiny ML}})},~m=1,2.
\ee
While the ML and MVU are $\Psi_{{\text{SMS}}}$-biased estimators for this case,
 it is shown {\em{analytically}} in \cite{song_uniform} that the U-V estimator,
\[\hat{\theta}_m^{({\mbox{\tiny U-V}})}=\hat{\theta}_m^{({\mbox{\tiny MVU}})}-\frac{1}{N+1}
\frac{\left(\hat{\theta}_k^{({\mbox{\tiny MVU}})}\right)^N}{\left(\hat{\theta}_m^{({\mbox{\tiny MVU}})}\right)^{N-1}}, 
\]
for $m,k=1,2$ and $m\neq k$,
satisfies
\beqna
{\rm{E}}_\thetavecsmall\left[\left(\hat{\theta}_m^{({\mbox{\tiny U-V}})}-{\theta}_m\right){\mathbf{1}}_{\{\Psi_{{\text{SMS}}}=m\}}
\right]
=0.
\eeqna
Thus, according to (\ref{unbias6}), the U-V estimator is an $\Psi_{{\text{SMS}}}$-unbiased estimator.
 Surprisingly, 
this
estimator is a function of the sufficient statistics of the two independent populations. 
In this case, the regularity conditions of  the likelihood function are not satisfied (e.g. \cite{Kay_estimation}, pp. 113-116); thus,
the   proposed $\Psi$-CRB for any selection rule $\Psi$, as well as the classical CRB itself, 
 do not exist.

The  PSMSE  of  the ML,  MVU, and U-V estimators 
 with the SMS  rule is
evaluated using $250,000$ Monte-Carlo simulations and presented in Fig.  \ref{MSSE_uniform}, 
for
$\theta_1=10$ and  $\theta_2=10.2$.
It can be seen that 
an  $\Psi_{{\text{SMS}}}$-unbiased estimator exists, $\hat{\thetavec}^{({\mbox{\tiny U-V}})}$, with a lower PSMSE    than the MMSE of the ML and MVU estimators for any number of samples, $N$.
 \begin{figure}[htb]
\vspace{-0.25cm}
\centerline{\psfig{figure=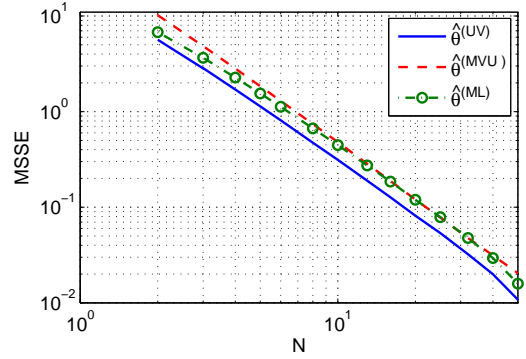,width=7.5cm}}
\caption{The  performance of the ML, MVU, and U-V estimators for estimation after parameter selection
 with independent  uniformly distributed populations and the SMS rule.}
\label{MSSE_uniform}
\vspace{-0.25cm}
\end{figure}

\subsection{Linear Gaussian model}
Consider the following observation model:
\beqna
\label{modelG}
\left\{\begin{array}{l}
y_1[n] = \theta_1 + w_1[n]\\
y_2[n] = \theta_2 +w_2[n]
\end{array}\right.,~~~n=0,\ldots,N-1,
\eeqna
 where  the normally distributed noise vectors, $\wvec[n]=[w_1[n],w_2[n]]^{\mbox{\tiny $T$}}$, $n=0,\ldots,N-1$,  are 
independent in time and space and have a
  zero mean  and  a known covariance matrix,
 \[\bsigma\define \left[\begin{array}{cc}\sigma_1^2&0\\
0 &\sigma_2^2
\end{array}\right]. \]

We assume the
 SMS rule, which selects the population with the largest
 sample mean,  \ie 
\be
\label{SMS_Gaussian}
\Psi_{{\text{SMS}}}=\arg \max_{m=1,2}\left\{\hat{\theta}_m^{({\mbox{\tiny ML}})}\right\},
\ee
where
\be
\label{ML_est}
\hat{\theta}_m^{({\mbox{\tiny ML}})}\define \frac{1}{N}\sum_{n=0}^{N-1} y_{m}[n],
\ee
is the ML estimator of $\theta_m$ for $m=1,2$.
According to (\ref{ML_est}), the  ML estimators are jointly Gaussian random variables with means $\theta_1$, $\theta_2$  and covariance matrix $\frac{1}{N}\bsigma$. Thus, 
for the SMS rule in (\ref{SMS_Gaussian}), the  probability of selecting population $m$ is:
\beqna
\label{prob_Gaussian}
\Pr(\Psi_{{\text{SMS}}}=m;\thetavec)=
\Pr(\hat{\theta}_m^{({\mbox{\tiny ML}})}-\hat{\theta}_k^{({\mbox{\tiny ML}})}>0;\thetavec)
=
 \Phi(\Delta_m)
,
\eeqna
 for $ m,k=1,2$, $m\neq k$,
where   $\Phi(\cdot)$
denotes the standard normal  cumulative distribution function (cdf),
 \[\Delta_m=\frac{\theta_m-\theta_k}{\sigma},~~~ m,k=1,2,~m\neq k,
\]
and
$\sigma^2\define \frac{\sigma_1^2+\sigma_2^2}{N}$.

\subsubsection{The $\Psi_{{\text{SMS}}}$-CRB}
It can be verified that for this case
\be
\label{diff_diff}
 \nabla_\thetavecsmall^2 \log f(\xvec;\thetavec)
=-N \bsigma^{-1}.
\ee
Therefore, by substituting (\ref{diff_diff}) in (\ref{JJJdef2}),
one obtains 
\beqna
\label{JJJdef2_Gaussian}
\Jmat_m(\thetavec,\Psi)=
 N \bsigma^{-1}
+ \nabla_\thetavecsmall^2
\log\Pr(\Psi=m;\thetavec),~m=1,2,
\eeqna
 for the   selection rule $\Psi$.
By substituting (\ref{JJJdef2_Gaussian}) in 
 (\ref{bound1}),
 the proposed $\Psi$-CRB is obtained.
Therefore, by using (\ref{prob_Gaussian}), the chain, and the product rules, it can be verified that
\beqna
\label{25}
 \nabla_\thetavecsmall^2 
\log\Pr(\Psi_{{\text{SMS}}}=m;\thetavec)
 =\frac{c(\Delta_{m}) }{\sigma^2}
\left[\begin{array}{rr}1 & -1
\\-1&1\end{array}\right],
\eeqna
 for all $ m=1,2$,
where
\beqna
\label{comp}
c(\Delta)\define 
-
\frac{\phi(\Delta )}{\Phi(\Delta)}
\Delta
-\frac{\phi^2(\Delta)}
{\Phi^2(\Delta)}
\eeqna
and $\phi(\cdot)$  denotes the standard normal pdf.
By substituting  (\ref{25})
in (\ref{JJJdef2_Gaussian}), 
one obtains
\beqna
\label{JJJdef2_Gaussian_inv}
\Jmat_m^{-1}(\thetavec,\Psi)=
\frac{1}{N}\left(\bsigma
-\frac{c(\Delta_m)} 
{N\sigma^2(c(\Delta_m)+1)}\Dmat\right)
,
\eeqna
where
\beqna
\Dmat\define
 \left[\begin{array}{cc} \sigma_1^4
&-\sigma_1^2\sigma_2^2 \\
-\sigma_1^2\sigma_2^2
 & \sigma_2^4
\end{array}\right].
\eeqna
By  substituting (\ref{JJJdef2_Gaussian_inv}) in
(\ref{bound_marginal}),  we obtain the  $\Psi_{\text{SMS}}$-CRB on the 
component-wise  PSMSE:
\beqna
\label{calculation2}
{\rm{E}}_{\thetavecsmall}\left[(\hat{\theta}_m-\theta_m)^2|\Psi_{\text{SMS}}=m\right]\geq
\frac{\sigma_m^2}{N} \zeta\left(\Delta_{m},\frac{\sigma_m^2 }{\sigma_1^2+\sigma_2^2}\right),
\eeqna
$m,k=1,2$ and $m\neq k$,
where
\beqna
\label{zeta_def}
\zeta\left(\Delta,\kappa\right)\define 1
-\frac{c(\Delta)} 
{c(\Delta)+1}\kappa.
\eeqna
 Finally,
 by substituting (\ref{prob_Gaussian})  and (\ref{calculation2}) 
 in (\ref{bound2}), 
 the  $\Psi_{\text{SMS}}$-CRB
 on the PSMSE is obtained:
\beqna
\label{bound2_Gaussian}
B^{(\Psi_{{\text{SMS}}})}(\thetavec)&=&\frac{\sigma_1^2}{N} \Phi\left(\Delta_1\right)
 \zeta\left(\Delta_1,\frac{\sigma_1^2 }{\sigma_1^2+\sigma_2^2}\right)
\nonumber\\&&
+\frac{\sigma_2^2}{N}\Phi\left(\Delta_2\right)
 \zeta\left(\Delta_2,\frac{\sigma_1^2 }{\sigma_1^2+\sigma_2^2}\right).
\eeqna
Similarly, the biased $\Psi$-CRB
 is obtained by substituting (\ref{JJJdef2_Gaussian_inv})
and  the gradient 
of the ML $\Psi$-biased from (\ref{bias1})-(\ref{bias2}) in (\ref{bound1_biased}).

It is well known that the   conventional CRB, on estimating the $m$th parameter $\theta_m$
without a selection stage,   is given by (e.g.  \cite{Kay_estimation} pp. 31-32)
$\frac{\sigma_m^2}{N}$.
Therefore, the component-wise $\Psi_{\text{SMS}}$-CRB in (\ref{calculation2}) is equal to the conventional marginal CRB  multiplied by a correction factor, $\zeta\left(\Delta,\sigma_1,\sigma_2\right)$, as defined in (\ref{zeta_def}).
The  correction factor
 is presented  in Fig. \ref{factor_Fig} versus $\Delta$ and $\sigma_1^2$ for $N=10$  and $\sigma_2=1$.
It can be seen that for  $\Delta<0$, the correction factor increases as $\Delta$ decreases, 
because this situation occurs when the order-relation  of the sample means is wrong, which is an indication of high estimation error.
Similarly, the correction factor increases as the variance $\sigma_1^2$ increases.
In contrast, 
for $c(\Delta)=0$, which occurs when $\Delta\gg 0$, 
the correction factor satisfies  $\zeta(\Delta,\sigma_1,\sigma_2)\rightarrow 1$ and thus,
 the component-wise $\Psi_{{\text{SMS}}}$-CRB  converges to the   CRB, \ie 
 the selection stage has only  minor influence on the estimation stage.
It should be noted, however, that the  CRB and $\Psi_{{\text{SMS}}}$-CRB are lower bounds on different performance
measures, \ie MSE and PSMSE, and on different groups of  estimators, \ie  mean-unbiased  and
$\Psi_{{\text{SMS}}}$-unbiased estimators, respectively.
\begin{figure}[htb]
\vspace{-0.25cm}
\centerline{\psfig{figure=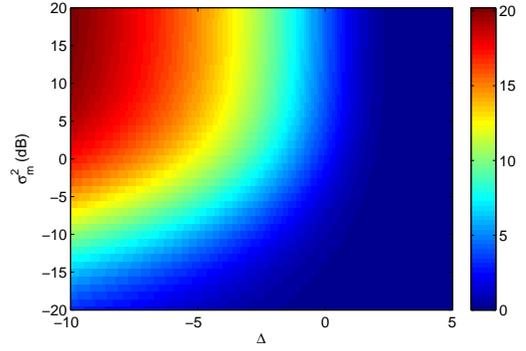,width=7cm}}
\caption{The multiplication factor, $\zeta(\Delta,\sigma_1,\sigma_2)$, for the linear Gaussian model with $N=10$ and $\sigma_2=1$.}
\label{factor_Fig}
\vspace{-0.25cm}
\end{figure}

\subsubsection{Estimation}
In \cite{Cohen_Sackrowitz} it is shown that  there is no  $\Psi_{{\text{SMS}}}$-unbiased estimator of $\theta_1$  and $\theta_2$. It is also shown that the ML estimator satisfies \cite{Cohen_Sackrowitz}
\be
\label{bias1}
{\rm{E}}_\thetavecsmall\left[\left.\hat{\theta}_m^{({\mbox{\tiny ML}})}-\theta_m\right|\Psi_{{\text{SMS}}}=m\right]
=\frac{\sigma_m^2}{N\sigma}\frac{\phi(\Delta_m )}{\Phi(\Delta_m)}
\geq 0
\ee
and
\be
\label{bias2}
{\rm{E}}_\thetavecsmall\left[\left.\hat{\theta}_k^{({\mbox{\tiny ML}})}-\theta_k\right|\Psi_{{\text{SMS}}}=m\right]
=-\frac{\sigma_k^2}{N\sigma}\frac{\phi(\Delta_m )}{\Phi(\Delta_m)}
\leq 0,
\ee
$m,k=1,2$, $m\neq k$. This result indicates that the ML tends to overestimate the parameter of the selected population and to underestimate the unknown parameter of the unselected population.

By using the model in (\ref{modelG}) and the selection probability in (\ref{prob_Gaussian}),
 we obtain
the following post-selection likelihood function for the SMS rule:
\beqna
\label{partial_Gaussian}
 \nabla_\thetavecsmall \log f(\xvec|\Psi_{{\text{SMS}}}=m;\thetavec)
=
N\bsigma^{-1}\left(\hat{\thetavec}^{({\mbox{\tiny ML}})}-\thetavec\right)
\nonumber\\
-\frac{ \phi(\Delta_m)}{ \sigma \Phi(\Delta_m)}
\left[\begin{array}{c}
{\mathbf{1}}_{\{\xvec\in {\mathcal{A}}_{1}\}}-{\mathbf{1}}_{\{\xvec\in {\mathcal{A}}_{2}\}}\\
{\mathbf{1}}_{\{\xvec\in {\mathcal{A}}_{2}\}}-{\mathbf{1}}_{\{\xvec\in {\mathcal{A}}_{1}\}}
\end{array}\right],
\eeqna
where $\hat{\thetavec}^{({\mbox{\tiny ML}})}=[\hat{\theta}_1^{({\mbox{\tiny ML}})},\hat{\theta}_2^{({\mbox{\tiny ML}})}]^{\mbox{\tiny $T$}}$
is defined in (\ref{ML_est}).
According to (\ref{PSML2}),
by equating the r.h.s. of (\ref{partial_Gaussian}) to zero we obtain 
the  PSML estimator for $\xvec\in {\mathcal{A}}_{m}$:
\beqna
\label{PSML_Gaussian}
\hat{\thetavec}^{({\mbox{\tiny PSML}})}=\hat{\thetavec}^{({\mbox{\tiny ML}})}\hspace{5cm}
\nonumber\\
-\frac{1}{ N \sigma } \frac{ \phi(\hat{\Delta}_m^{({\mbox{\tiny PSML}})})}{ \Phi(\hat{\Delta}_m^{({\mbox{\tiny PSML}})})}\left[\begin{array}{c}\sigma_1^2\left(
{\mathbf{1}}_{\{\xvec\in {\mathcal{A}}_{1}\}}-{\mathbf{1}}_{\{\xvec\in {\mathcal{A}}_{2}\}}\right)\\
\sigma_2^2\left({\mathbf{1}}_{\{\xvec\in {\mathcal{A}}_{2}\}}-{\mathbf{1}}_{\{\xvec\in {\mathcal{A}}_{1}\}}\right)
\end{array}\right],
\eeqna
 where
 \[\hat{\Delta}_m^{({\mbox{\tiny PSML}})}=\frac{\hat{\theta}_{m}^{({\mbox{\tiny PSML}})}-\hat{\theta}_{k}^{({\mbox{\tiny PSML}})}}{\sigma},~~~ m,k=1,2,~m\neq k,
\]
for any $\xvec\in {\mathcal{A}}_{m}$.
 It can be seen that
 as  $\sigma^2$ increases the correction term on the r.h.s. of  (\ref{PSML_Gaussian}) 
becomes insignificant and the PSML estimator approaches the ML estimator.

The solution of (\ref{PSML_Gaussian}) can be found by an exhaustive search over
 $\hat{\thetavec}^{({\mbox{\tiny PSML}})}$
or by using the iterative  methods from Section \ref{itr_section}.
It can be verified that the Newton-Raphson and post-selection Fisher scoring coincide in this case,
where the $i$th iteration of the Newton-Raphson PSML (NR-PSML) is obtained
by
 substituting  (\ref{partial_Gaussian}) and 
 (\ref{JJJdef2_Gaussian_inv}) 
in (\ref{fisher_scoring}).
Similarly,
by substituting  (\ref{prob_Gaussian}), (\ref{diff_diff}), and (\ref{25}) in (\ref{by_parts_max}),  the MBP estimator is obtained:
\beqna
\label{g_ML_Gaussian2}
\hat{\thetavec}^{(i+1)}=
\hat{\thetavec}^{({\mbox{\tiny ML}})}\hspace{6cm}\nonumber\\-\frac{\phi(\hat{\Delta}^{(i)})}{N \sigma \Phi(\hat{\Delta}^{(i)})}
\left[\begin{array}{c}\sigma_1^2 \left(
{\mathbf{1}}_{\{\xvec\in {\mathcal{A}}_{1}\}}-{\mathbf{1}}_{\{\xvec\in {\mathcal{A}}_{2}\}}\right)\\
\sigma_2^2 \left({\mathbf{1}}_{\{\xvec\in {\mathcal{A}}_{2}\}}-{\mathbf{1}}_{\{\xvec\in {\mathcal{A}}_{1}\}}\right)
\end{array}\right],
\eeqna
which coincides with the results in \cite{Bebu_Luta_Dragalin_2010} for the Gaussian case.

The bias and PSMSE of  the ML, NR-PSML, and MBP estimators  
 with the SMS  rule are
evaluated using $20,000$ Monte-Carlo simulations, 
and the results are presented in Figs. \ref{bias_Gaussian} and \ref{MSSE_Gaussian}, respectively,
for
$\theta_1=0$, $\theta_2=0.1$, $\sigma_1^2=1$, and $\sigma_2^2=0.1$.
The PSMSE performance is compared to the $\Psi$-CRB from (\ref{bound2_Gaussian})
and the biased $\Psi$-CRB.
It can be seen that the PSML methods have a lower $\Psi$-bias and lower PSMSE than the ML estimator and that
the MBP estimator has the best performance in both terms.
Since  no $\Psi$-unbiased estimator exists, the  $\Psi_{{\text{SMS}}}$-CRB is higher than the  actual PSMSE values, but the biased $\Psi$-CRB is a valid bound for any $N$.
However, it can be seen that the $\Psi_{{\text{SMS}}}$-CRB gives an indication of the performance behavior and that asymptotically it is attained by the PSML estimator and coincides with the biased $\Psi$-CRB.
 \begin{figure}[htb]
\vspace{-0.25cm}
\centerline{\psfig{figure=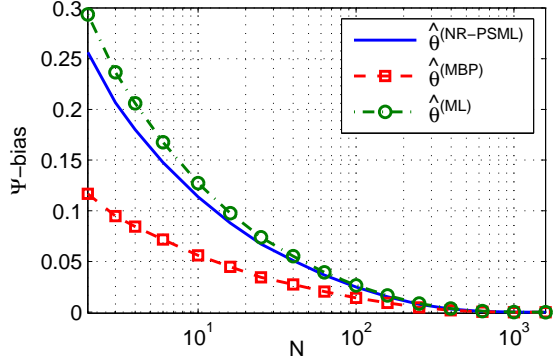,width=8cm}}
\caption{The $\Psi_{{\text{SMS}}}$-bias of the ML, NR-PSML, and MBP estimators  
for estimation after parameter selection  with independent  Gaussian distributed populations and the SMS  rule.}
\label{bias_Gaussian}
\vspace{-0.25cm}
\end{figure}
 \begin{figure}[htb]
\vspace{-0.25cm}
\centerline{\psfig{figure=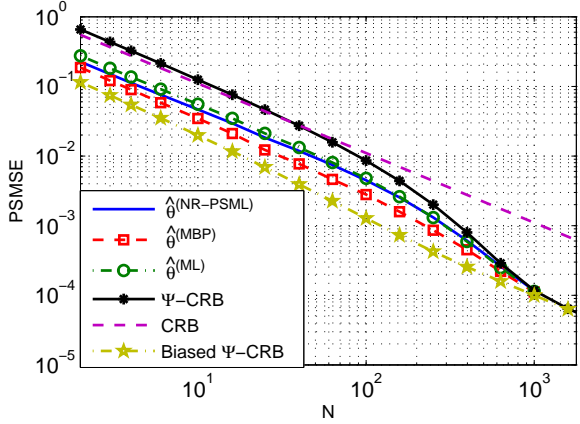,width=8.5cm}}
\caption{The  PSMSE of the ML, NR-PSML, and MBP estimators  
and the $\Psi$-CRB
for estimation after parameter selection
with independent Gaussian distributed populations and the SMS rule.}
\label{MSSE_Gaussian}
\vspace{-0.25cm}
\end{figure}

\subsection{Exponential distribution}
Consider the following observation model:
\beqna
\label{model_exponential}
f_{m}(y_m[n];\theta_m)=\left\{
\begin{array}{lr}\frac{1}{\theta_m}e^{-\frac{y_m[n]}{\theta_m}}&0<y_m[n]\\
0&{\text{otherwise}}
\end{array}
\right.,~m=1,2,
\eeqna
for all $ n=0,\ldots,N-1$, where the parameters $\theta_m>0$ for all  $ m=1,2$ are unknown.
The populations are assumed to be independent. 
In this problem, the SMS rule selects the population with the largest sample-mean,  \ie 
\be
\label{SMS_exponential}
\Psi_{{\text{SMS}}}=\arg \max_{m=1,2}\left\{\hat{\theta}_m^{({\mbox{\tiny ML}})}\right\},
\ee
where
\be
\label{ML_est_exponential}
\hat{\theta}_m^{({\mbox{\tiny ML}})}\define \frac{1}{N}\sum_{n=0}^{N-1} y_{m}[n]
\ee
is the ML estimator of $\theta_m$ for $m=1,2$.
The   probability of selecting population $m$ via  the SMS rule is given by the  negative binomial cdf \cite{Vellaisamy_Sharma_1988}:
\beqna
\label{prob_exponential}
\Pr(\Psi_{{\text{SMS}}}=m;\thetavec)=
\sum_{j=0}^{N-1}
\left(\begin{array}{c}N+j-1\\j \end{array}\right)
q_m^N (1-q_m)^j
\eeqna
for any $ \theta_m,\theta_k>0$,
where 
$q_m\define \frac{ \theta_m  }{ \theta_m+\theta_k}$.

\subsubsection{Estimation}
The ML estimator from (\ref{ML_est_exponential}) is  an $\Psi_{{\text{SMS}}}$-biased estimator in this case, since
  \cite{Vellaisamy_Sharma_1988}
\beqna
\label{expectation_m2}
{\rm{E}}_\thetavecsmall\left[\left. \hat{\theta}_m^{({\mbox{\tiny ML}})}-\theta_m
\right|\Psi_{{\text{SMS}}}=m\right]
=\theta_m \alpha_m
\eeqna
for any $ \theta_m,\theta_k>0$, where
$ m,k=1,2$, $m\neq k$
and
\[\alpha_m\define 
\frac{ 
1}{\Pr(\Psi_{{\text{SMS}}}=m;\thetavec)}\left(\begin{array}{c}2N-1\\N \end{array}\right)
q_m^N(1-q_m)^N.
\]
In addition, by using (\ref{expectation_m2}) it can be verified that
\beqna
\label{expectation_k2}
{\rm{E}}_\thetavecsmall\left[\left. \hat{\theta}_k^{({\mbox{\tiny ML}})}
-\theta_k
\right|\Psi_{{\text{SMS}}}=m\right]
=-\theta_k \alpha_m
\eeqna
for any $ \theta_m,\theta_k>0$, where
$m,k=1,2$, $m\neq k$.
The $\Psi_{\text{SMS}}$-biases in (\ref{expectation_m2}) and (\ref{expectation_k2}) are positive and negative, respectively,  thus they tend to overestimate  the parameter of the selected population and 
underestimate that of the unselected
 one.
As an alternative to the ML estimator,  the following U-V estimator is proposed in \cite{Vellaisamy_Sharma_1988}: 
\be
\label{UV_exponential}
\hat{\theta}_m^{({\mbox{\tiny U-V}})}=\hat{\theta}_m^{({\mbox{\tiny ML}})}
-
\frac{\left(\hat{\theta}_k^{({\mbox{\tiny ML}})}\right)^N}{\left(\hat{\theta}_m^{({\mbox{\tiny ML}})}\right)^{N-1}}, ~m,k=1,2,~~m\neq k.
\ee
It is also shown that
$
{\rm{E}}_\thetavecsmall\left[\left(\hat{\theta}_m^{({\mbox{\tiny U-V}})}-{\theta}_m\right){\mathbf{1}}_{\{\Psi_{{\text{SMS}}}=m\}}
\right]
=0
$
and thus, according to (\ref{unbias6}), the U-V estimator  is an $\Psi_{{\text{SMS}}}$-unbiased estimator.

In the following, we derive the PSML estimator for $\Psi_{{\text{SMS}}}=m$.
The results for $\Psi_{{\text{SMS}}}=k$, $k\neq m$ are straightforward.
For the sake of simplicity, the elements of $\thetavec$ are reordered such that
the first element is the selected one, \ie 
$\thetavec=[\theta_m,\theta_k]^{\mbox{\tiny $T$}}$.
The  PSML estimator   from (\ref{PSML}) maximizes
the post-selection likelihood, which is given in this case by
\beqna
\label{cond_like_exp1}
 &&\hspace{-1cm}\log f(\xvec|\Psi_{{\text{SMS}}}=m;\thetavec)
\nonumber\\
&=&-N \log \theta_m
-N\log \theta_k
-  \frac{N \hat{\theta}_m^{({\mbox{\tiny ML}})}}{\theta_m}
-  \frac{N \hat{\theta}_k^{({\mbox{\tiny ML}})}}{\theta_k}
\nonumber\\&&
-\log\left(
\sum_{j=0}^{N-1}
\left(\begin{array}{c}N+j-1\\j \end{array}\right) q_m^N (1-q_m)^j\right)
\eeqna
for any $ \theta_m,\theta_k>0$ and $\xvec\in {\mathcal{A}}_{m}$.
By using (\ref{cond_like_exp1}), it can be verified that
the 
gradient vector of $f(\xvec|\Psi_{{\text{SMS}}}=m;\thetavec)$  w.r.t. $\thetavec$  is given by
\beqna
\label{gradient}
 \nabla_\thetavecsmall \log f(\xvec|\Psi_{{\text{SMS}}}=m;\thetavec)
=
N\left[
\begin{array}{c}
 \frac{ \hat{\theta}_m^{({\mbox{\tiny ML}})}-\theta_m+\theta_m f(q_m)}{\theta_m^2}
\\
\frac{ \hat{\theta}_k^{({\mbox{\tiny ML}})}-\theta_k-\theta_k f(q_m)}{\theta_k^2}
\end{array}
\right],
\eeqna
for any $ \theta_m,\theta_k>0$,
$ m,k=1,2$, and $m\neq k$,
where
\be
f(q_m)\define (1-q_m)(q_m h(q_m)-1)
\ee
and
\be
h(q_m)
\define \frac{\sum_{j=0}^{N-2}
\left(\begin{array}{c}N+j\\j \end{array}\right)
q_m^N (1-q_m)^{j}}{\Pr(\Psi_{\text{SMS}}=m;\thetavec)},~m=1,2.
\ee
According to (\ref{PSML2}),
by equating the r.h.s. of (\ref{gradient}) to zero we obtain 
the  PSML estimator for $\xvec\in {\mathcal{A}}_{m}$:
\beqna
\label{derivative_1}
\left[\begin{array}{c}\hat{\theta}_{m}^{({\mbox{\tiny PSML}})}
\\\hat{\theta}_{k}^{({\mbox{\tiny PSML}})}\end{array}\right]
=
\left[\begin{array}{c}\frac{\hat{\theta}_m^{({\mbox{\tiny ML}})}}
{1
-f(\hat{q}_m)  }
\\
 \frac{ \hat{\theta}_k^{({\mbox{\tiny ML}})}}{1
+f(\hat{q}_m) }
\end{array}\right],
\eeqna
where
 $\hat{q}_m\define \frac{ \hat{\theta}_{m}^{({\mbox{\tiny PSML}})}  }{ \hat{\theta}_{m}^{({\mbox{\tiny PSML}})}+\hat{\theta}_{k}^{({\mbox{\tiny PSML}})}}$
and $\xvec \in {\mathcal{A}}_{m}$.
Equation
(\ref{derivative_1}) 
implies 
 that the ratios 
$ \frac{ \hat{\theta}_{m}^{({\mbox{\tiny PSML}})}}{ \hat{\theta}_m^{({\mbox{\tiny ML}})}}$
and
$ \frac{ \hat{\theta}_{k}^{({\mbox{\tiny PSML}})}}{ \hat{\theta}_k^{({\mbox{\tiny ML}})}}$ are only functions of
the statistic $ \frac{\hat{\theta}_k^{({\mbox{\tiny ML}})}}{\hat{\theta}_m^{({\mbox{\tiny ML}})}}$.
That is, the PSML estimator is a function of the ML estimator multiplied by a correction factor, which is a function of the ML estimators' ratio. 

The solution to (\ref{derivative_1}) can be found for the general case by an exhaustive search over
 $\hat{\thetavec}^{({\mbox{\tiny PSML}})}$
or by using the iterative  methods from Section \ref{itr_section}.
For example, 
by substituting  (\ref{prob_exponential}), (\ref{diff_diff}), and (\ref{25}) in (\ref{by_parts_max}),  the $i$th iteration of the  MBP method is obtained.

\subsubsection{$\Psi_{\text{SMS}}$-CRB}
The PSFIM can be obtained 
by using the derivative of (\ref{gradient}), applying the expectation operator, and
 using 
(\ref{expectation_m2}) and (\ref{expectation_k2}). 
Then, the  $\Psi_{\text{SMS}}$-CRB is obtained by substituting the PSFIM
in  (\ref{bound2}).
The explicit $\Psi_{\text{SMS}}$-CRB is omitted from this paper due to space limitation.

\subsubsection{$\Psi_{\text{SMS}}$-efficiency}
For the special case of 
 $N=1$,   it can be shown that  the  $\Psi_{\text{SMS}}$-CRB  is given by
\beqna
\label{bound2_exp}
B^{(\Psi_{{\text{SMS}}})}(\thetavec)
=\frac{ \theta_m^3 + \theta_k^3}{ \theta_m+\theta_k }.
\eeqna
The PSML estimator in this case is given by
\beqna
\label{SML_m1}
\left[\begin{array}{c}\hat{\theta}_{m}^{({\mbox{\tiny PSML}})}
\\\hat{\theta}_{m}^{({\mbox{\tiny PSML}})}
\end{array}\right]
=   
\left[\begin{array}{c}y_m[0]
-y_k[0]
\\y_k[0]\frac{   y_m[0]
-y_k[0]}{ y_m[0]
-2y_k[0]}
\end{array}\right]
\eeqna
for any  $y_m[0]\geq 
y_k[0]$ and $y_m[0]\neq 
2y_k[0]$. 
For $y_k[0]\geq 
y_m[0]$, we can change the roles of $\theta_m$ and $\theta_k$ to obtain the PSML estimator.
It can be seen that for $N=1$, the PSML and U-V estimators from  (\ref{SML_m1}) and (\ref{UV_exponential}), respectively, of the selected parameter coincides, \ie
$\hat{\theta}_{m}^{({\mbox{\tiny PSML}})}=\hat{\theta}_m^{({\mbox{\tiny U-V}})}$.
Thus, the PSML estimator is an $\Psi_{\text{SMS}}$-unbiased estimator for $N=1$.
In addition, we can verify analytically that the PSMSE of $\hat{\thetavec}^{({\mbox{\tiny PSML}})}$ attains the $\Psi_{\text{SMS}}$-CRB from (\ref{bound2_exp}).
Thus, $\hat{\thetavec}^{({\mbox{\tiny PSML}})}$ is an 
 $\Psi_{\text{SMS}}$-efficient estimator  for $N=1$.

The PSMSE performance of the estimators $\hat{\thetavec}^{({\mbox{\tiny ML}})}$ and
$\hat{\thetavec}^{({\mbox{\tiny PSML}})}$
 with the SMS  rule are
evaluated using $100,000$ Monte-Carlo simulations 
and are compared with the $\Psi_{\text{SMS}}$-CRB for $N=1$ and $\theta_1=5$.
The results are presented in Figs.  \ref{bias_exp} and \ref{PSMSE_exp}.
It can be seen that the PSML estimator is an $\Psi$-unbiased estimator  and has a lower PSMSE than the ML estimator.
Moreover, 
 the  $\Psi_{{\text{SMS}}}$-CRB is  achievable by  $\hat{\thetavec}^{({\mbox{\tiny PSML}})}$, which is an  $\Psi_{\text{SMS}}$-efficient estimator  in this case.
 \begin{figure}[htb]
\vspace{-0.25cm}
\centerline{\psfig{figure=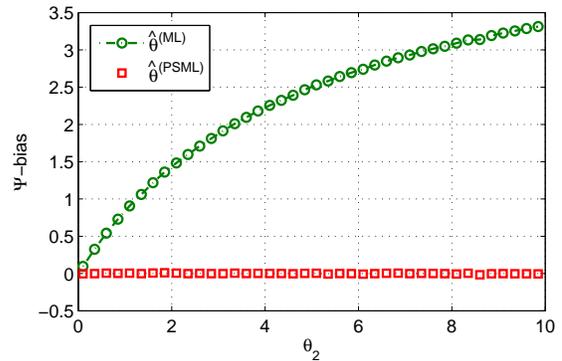,width=8cm}}
\caption{The $\Psi_{{\text{SMS}}}$-bias of the ML and PSML estimators  
for estimation after parameter selection with independent exponential distributed populations and the SMS rule.}
\label{bias_exp}
\vspace{-0.25cm}
\end{figure}
 \begin{figure}[htb]
\vspace{-0.25cm}
\centerline{\psfig{figure=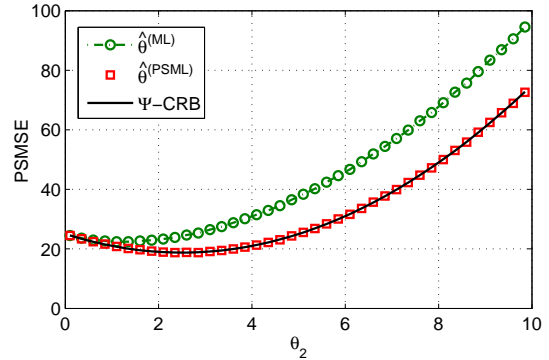,width=8cm}}
\caption{The  PSMSE of the ML and PSML estimators and the $\Psi$-CRB  for estimation after parameter selection
with independent exponential distributed populations and the SMS rule.}
\label{PSMSE_exp}
\vspace{-0.25cm}
\end{figure}

\section{Conclusion}
 \label{diss}
 In this paper, 
the concept of non-Bayesian  estimation after parameter selection is introduced 
and the $\Psi$-unbiasedness in the Lehmann sense is defined,
for  arbitrary data-driven parameter selection rules.
We derive a  Cram$\acute{\text{e}}$r-Rao-type bound 
for the selected deterministic parameters.
 Unlike the conventional CRB,  the proposed $\Psi$-CRB provides a valid bound in estimation after parameter selection problems.
The PSML estimator is proposed and its properties and practical implementations aspects are discussed.
In particular, 
it is proved that if there exists  an $\Psi$-efficient estimator,
 then it is produced by the PSML estimator.
The new paradigm opens a wide range of interesting directions, such as  
multistage procedures that involve active learning and sequential data sampling.

\section*{\normalsize{Appendix A: Proof of Proposition \ref{unbiasedness_prop}}}
\label{Appendix A}
In this Appendix, we prove Proposition \ref{unbiasedness_prop}
in a similar way to the proof
of mean-unbiasedness under a conventional squared error cost function (Page 11 in \cite{Lehmann}).
By substituting the  PSSE cost function from (\ref{cost_def})
in (\ref{defdef})
 the Lehmann-unbiasedness  condition  is given by 
\beqna
\label{def_per_general}
{\rm{E}}_{\thetavecsmall}\left[\sum_{m=1}^M
(\hat{\theta}_m-\eta_m)^2{\mathbf{1}}_{\{\Psi=m\}}\right]\geq \hspace{2.75cm}
\nonumber\\
  {\rm{E}}_{\thetavecsmall}\left[\sum_{m=1}^M
(\hat{\theta}_m-\theta_m)^2{\mathbf{1}}_{\{\Psi=m\}}
\right],~~~\forall \thetavec,\etavec\in {\mathbb{R}}^M,
\eeqna
where $\etavec=[\eta_1,\ldots,\eta_M]^{\mbox{\tiny $T$}}$ is an arbitrary vector.
The condition in (\ref{def_per_general}) can be rewritten as
\beqna
\label{def_per_general2}
{\rm{E}}_{\thetavecsmall}\left[\sum_{m=1}^M
(\hat{\theta}_m-\theta_m+\theta_m-\eta_m)^2{\mathbf{1}}_{\{\Psi=m\}}\right]\hspace{1.5cm}
\nonumber\\
 \geq  {\rm{E}}_{\thetavecsmall}\left[\sum_{m=1}^M
(\hat{\theta}_m-\theta_m)^2{\mathbf{1}}_{\{\Psi=m\}}
\right],~~~\forall \thetavec,\etavec\in {\mathbb{R}}^M.
\eeqna
By using the linearity of  the expectation operator
and the fact that  $\thetavec$ and $\etavec$ are deterministic vectors, 
it can be verified that
(\ref{def_per_general2}) is equivalent to
\beqna
\label{def_per_general3}
\sum_{m=1}^M(\theta_m-\eta_m){\rm{E}}_{\thetavecsmall}\left[
(\hat{\theta}_m-\theta_m){\mathbf{1}}_{\{\Psi=m\}}\right]
\nonumber\\
 \geq  -\sum_{m=1}^M(\theta_m-\eta_m)^2 \Pr(\Psi=m;\thetavec),
\eeqna
$\forall \thetavec,\etavec\in {\mathbb{R}}^M$,
where we used ${\rm{E}}_{\thetavecsmall}\left[{\mathbf{1}}_{\{\Psi=m\}}\right]=\Pr(\Psi=m;\thetavec)$.
.\\
{\bf{Sufficient condition - }} 
It can be verified that  if (\ref{unbias6})  holds,
 then the inequality in (\ref{def_per_general3}) holds since the r.h.s. of (\ref{def_per_general3}) is nonpositive.\\
{\bf{Necessary condition - }} The necessity of  (\ref{unbias6})  is proven by
using specific choices of $\etavec$.
 By substituting  
$\eta_m=\theta_m$, for all $m=1,\ldots,M$, $m\neq l$, and $\eta_l=\theta_l+\varepsilon$
in  (\ref{def_per_general3}), 
we obtain the following necessary condition:
\be
\label{def_per444}
\varepsilon{\rm{E}}_{\thetavecsmall}\left[
(\hat{\theta}_l-\theta_l){\mathbf{1}}_{\{\Psi=l\}}\right]
 \geq  -\varepsilon^2  \Pr(\Psi=l;\thetavec)
\ee
for any $\thetavec \in{\mathbb{R}}^M$, $\varepsilon\in{\mathbb{R}}$.
Since $\varepsilon$ can be either positive or negative and $\Pr(\Psi=l;\thetavec)\geq 0$, the condition in 
(\ref{def_per444}) implies (\ref{unbias6}) for any $l=1,\ldots,M$.
In addition, since
\beqna
{\rm{E}}_{\thetavecsmall}[ (\hat{\theta}_m-\theta_m){\mathbf{1}}_{\{\Psi=m\}}]
= {\rm{E}}_{\thetavecsmall}[\hat{\theta}_m-\theta_m|\Psi=m]\Pr(\Psi=m;\thetavec ), \nonumber
\eeqna 
$ m=1,\ldots,M$,  
 then,  for any $ m=1,\ldots,M$ such   that $\Pr(\Psi=m;\thetavec )\neq 0$, 
the condition in (\ref{unbias6}) 
 is equivalent  to (\ref{10}).

\section*{Appendix B. Proof of Lemma \ref{lemma1}}
In this Appendix, two alternative formulations of the PSFIM are derived.
Similar derivations can be found in \cite{Chaumette2005}
in the context of  post-detection estimation.
By using  (\ref{Bayes}), one obtains
\beqna
\label{log_divide}
\nabla_\thetavecsmall \log f(\xvec|\Psi=m;\thetavec)\hspace{4cm}\nonumber\\
= \nabla_\thetavecsmall \log f(\xvec;\thetavec)
-\nabla_\thetavecsmall \Pr(\Psi=m;\thetavec),
\eeqna
$\forall \xvec\in {\mathcal{A}}_{m}$,
 and by substituting (\ref{log_divide}) in
(\ref{JJJdef}), one obtains
\beqna
\label{JJJdef_re}
\Jmat_m(\thetavec,\Psi)= \hspace{5.75cm}\nonumber\\ {\rm{E}}_{\thetavecsmall}\left[\left.
 \nabla_\thetavecsmall \log f(\xvec;\thetavec)
\nabla_\thetavecsmall^T\log f(\xvec;\thetavec) \right|\Psi=m
\right]\hspace{0.9cm}
\nonumber\\- \nabla_\thetavecsmall \log\Pr(\Psi=m;\thetavec)
{\rm{E}}_{\thetavecsmall}\left[\left.\nabla_\thetavecsmall^T \log f(\xvec;\thetavec) \right|\Psi=m
\right]
\nonumber\\-{\rm{E}}_{\thetavecsmall}\left[\left.
 \nabla_\thetavecsmall \log f(\xvec;\thetavec)\right|\Psi=m
\right]\nabla_\thetavecsmall^T \log \Pr(\Psi=m;\thetavec)
\nonumber\\+
 \nabla_\thetavecsmall \log \Pr(\Psi=m;\thetavec)
\nabla_\thetavecsmall^T\log\Pr(\Psi=m;\thetavec).\hspace{0.9cm}
\eeqna
Since ${\mathcal{A}}_{m}$ is independent of $\thetavec$
and by using regularity condition \ref{cond2},
it can be noticed that
\beqna
\label{nine}
&&\hspace{-2cm}\nabla_\thetavecsmall \log \Pr(\Psi=m;\thetavec)=
\frac{\nabla_\thetavecsmall \Pr(\Psi=m;\thetavec)}{\Pr(\Psi=m;\thetavec)}
\nonumber\\
&=&\frac{\nabla_\thetavecsmall \int_{{\mathcal{A}}_{m}}f(\xvec;\thetavec)\ud \xvec}{\Pr(\Psi=m;\thetavec)}
=\frac{\int_{{\mathcal{A}}_{m}} \nabla_\thetavecsmall f(\xvec;\thetavec)\ud \xvec}{\Pr(\Psi=m;\thetavec)}
\nonumber\\
&=&{\rm{E}}_{\thetavecsmall}\left[\left.
\nabla_\thetavecsmall \log f(\xvec;\thetavec) \right|\Psi=m
\right].
\eeqna
Substitution of (\ref{nine}) in (\ref{JJJdef_re}) results in  (\ref{JJJdef1.5}).

In addition, under the assumption that 
the integral
 $\int_{{\mathcal{A}}_{m}} f(\xvec|\Psi=m ;\thetavec){\ud}\xvec$
 can be twice differentiated under the integral
sign, it is known   that (Lemma 2.5.3 in \cite{point_est}): 
\[{\rm{E}}_{\thetavecsmall}\left[\left.
\nabla_\thetavecsmall \log f(\xvec|\Psi=m ;\thetavec)\right|\Psi=m\right]=0\]
for any $ \thetavec \in{\mathbb{R}}^M$.
Therefore, 
by using the product rule twice 
we obtain
\beqna
\label{thanks1}
\hspace{-0.5cm}\Jmat_m(\thetavec,\Psi)&=& {\rm{E}}_{\thetavecsmall}\left[\left.
\nabla_\thetavecsmall \log f(\xvec|\Psi=m;\thetavec)\right.\right.
\nonumber\\&&\times
\left.\left.
\nabla_\thetavecsmall^T\log f(\xvec|\Psi=m;\thetavec) \right|\Psi=m
\right]
 \nonumber\\
 &=&\nabla_\thetavecsmall {\rm{E}}_{\thetavecsmall}\left[\left.
\nabla_\thetavecsmall^T\log f(\xvec|\Psi=m ;\thetavec)\right|\Psi=m\right] \nonumber\\
 &&-{\rm{E}}_{\thetavecsmall}\left[\left.
\nabla_\thetavecsmall^2 \log f(\xvec|\Psi=m;\thetavec) \right|\Psi=m\right]
\nonumber\\&=&
 -{\rm{E}}_{\thetavecsmall}\left[\left.
\nabla_\thetavecsmall^2 \log f(\xvec|\Psi=m;\thetavec) \right|\Psi=m\right].
\eeqna
By substituting  (\ref{log_divide}) in (\ref{thanks1}), we obtain (\ref{JJJdef2}).

\section*{Appendix C: Numerical PSML estimation method}
In some instances, the post-selection likelihood function and its gradient are intractable, so that finding the PSML, even by the iterative methods in 
(\ref{iterative_NR2})
and
(\ref{fisher_scoring2}),
 may be difficult. 
In these cases, we can  use the previous estimator $\hat{\thetavec}^{(i)}$ to construct a nonparametric estimator 
 of
$\gvec_m(\thetavec)\define \nabla_\thetavecsmall \log \Pr(\Psi=m;\thetavec) 
$
at $\thetavec=\hat{\thetavec}^{(i)}$
 from simulated realizations of the observation model and 
then, to substitute them 
 in
(\ref{iterative_NR2}) or
(\ref{fisher_scoring2}).
The resulting iterative PSML (IPSML) algorithm
is described  in Table \ref{Box1}.
\texttt{ 
\begin{table}[h]
\vspace{-0.25cm}
\caption{The IPSML  algorithm}
\begin{tabular}{|p{8cm}|}
\hline\\[3pt]
{\bf{Initialization:}} 
Fix $i=0$ and 
set
 the temporary estimator 
$\hat{\thetavec}^{(0)}
=\hat{\thetavec}^{({\mbox{\tiny ML}})}$.\\
{\bf  Main iteration:} 
Increment $i\rightarrow i+1$ and apply
\begin{enumerate}
\item {\bf{Empirical gradient:}}
 For any $l=1,\ldots,M$:
\begin{enumerate}
\item Generate  data $\tilde{\xvec}_k^+$ and $\tilde{\xvec}_k^-$ according to the pdf's
$f(\cdot;\hat{\thetavec}^{(i)}+\frac{\Delta}{2}\evec_l)$ and
$f(\cdot;\hat{\thetavec}^{(i)}-\frac{\Delta}{2}\evec_l)$ for $k=1,\ldots,K$.
\item Evaluate the empirical partial derivative
 by:
\beqna
\left[\hat{\gvec}_m\left(\thetavec=\hat{\thetavec}^{(i)}\right)\right]_l
\approx \frac{ 1}{\Delta}\times\hspace{3
cm}
\nonumber\\\hspace{-0.5cm}
\log\left(\frac{1}{K}\sum\limits_{k=1}^K {\mathbf{1}}_{\Psi(\tilde{\xvec}_k^+)=m}\right)
 -\log\left(\frac{1}{K}\sum\limits_{k=1}^K {\mathbf{1}}_{\Psi(\tilde{\xvec}_k^-)=m}\right)
.\nonumber
\eeqna
\item {\bf{Update gradient:}} 
\beqna
\label{cond_log_approx}
\left. \nabla_\thetavecsmall \log f(\xvec|\Psi=m;\thetavec)\right|_{\thetavecsmall=\hat{\thetavecsmall}^{(i)}}\hspace{1.5cm}\nonumber\\
\approx\left. \nabla_\thetavecsmall f(\xvec;\thetavec)\right|_{\thetavecsmall=\hat{\thetavecsmall}^{(i)}}
-\hat{\gvec}_m(\thetavec=\hat{\thetavec}^{(i)}).
\eeqna
\item {\bf{Update estimation:}}  Substitute (\ref{cond_log_approx}) in (\ref{iterative_NR2}) or
(\ref{fisher_scoring2}).
\end{enumerate}
\item {\bf{Stopping rule:}} iterate till convergence
\end{enumerate}
{\bf{Output:}}  $\hat{\thetavec}^{({\mbox{\tiny IPSML}})}=\hat{\thetavec}^{(i)}$.
\\[3pt]
 \hline
\end{tabular}
\label{Box1}
\vspace{-0.25cm}
\end{table}
}

\bibliographystyle{IEEEtran}
\bibliography{errorbound5}

\end{document}

%% file: myshorts.tex
\newcommand{\evec}{{\bf{e}}}

\newcommand{\yvec}{{\bf{y}}}

\newcommand{\wvec}{{\bf{w}}}
\newcommand{\xvec}{{\bf{x}}}

\newcommand{\gvec}{{\bf{g}}}

\newcommand{\etavec}{{\bf{\eta}}}

\newcommand{\zerovec}{{\bf{0}}}

\newcommand{\Amat}{{\bf{A}}}
\newcommand{\Bmat}{{\bf{B}}}

\newcommand{\Dmat}{{\bf{D}}}

\newcommand{\Hmat}{{\bf{H}}}
\newcommand{\Jmat}{{\bf{J}}}

\newcommand{\define}{\stackrel{\triangle}{=}}

\def\bsigma{{\mbox{\boldmath $\Sigma$}}}

\def\etavec{{\mbox{\boldmath $\eta$}}}

\def\thetavec{{\mbox{\boldmath $\theta$}}}

\def\thetavecsmall{{\mbox{\boldmath {\scriptsize $\theta$}}}}

\newcommand{\be}{\begin{equation}}
\newcommand{\ee}{\end{equation}}
\newcommand{\beqna}{\begin{eqnarray}}
\newcommand{\eeqna}{\end{eqnarray}}
